\documentclass[10pt]{article}

\usepackage[T1]{fontenc}										% font encoding ... definitely need this
\usepackage{lmodern}												% loads latin modern fonts ... not sure if needed but doesn't hurt to include it
\usepackage{microtype}											% improves spacing
\usepackage[LGRgreek,frenchmath]{mathastext}% sets default math font to the default roman text font ... option [frenchmath] for lower case Roman fonts in italics
\usepackage[mathscr]{eucal}                 % gives you \mathscr{} font
\usepackage[sans]{dsfont}										% gives you \mathds{} font ... option [sans] for sans serif
\usepackage{bbold}													% gives you sans serif \mathbb{} font  ... may need to compile with dvi or ps first ... not sure
\usepackage{bbm}														% gives you various \mathbbm{} and \mathbbmss{} and \mathbbmtt{} fonts
\usepackage{graphicx}                      	% need for figures
\usepackage{mdwlist}												% need for itemize* (compact itemize)
\usepackage{natbib}    											% need for bibliography
\setcitestyle{authoryear}
\usepackage[dvipsnames]{xcolor}							% gives you color definitions
\usepackage[plainpages=false, pdfpagelabels]{hyperref} % enables and customizes hyperlinks
	\hypersetup{
		colorlinks   = true,
		citecolor    = RoyalBlue, %NavyBlue,
		linkcolor    = RubineRed, %Maroon,
		urlcolor     = Turquoise
	}
\usepackage[paperwidth=8.5in,paperheight=11.00in,top=1.25in, bottom=1.25in, left=1.00in, right=1.00in]{geometry}
\usepackage{mathtools}                      % need for `show only references'
\mathtoolsset{showonlyrefs=true}            % only equations which are labeled AND referenced will numbered...use \eqref{} instead of (\ref{})
\usepackage{amsmath}
\usepackage{amssymb}
\usepackage{amsfonts}
\usepackage{amsthm}                         % need for theorem-proof environment
\allowdisplaybreaks                         % allows page breaks for long equations...you can prevent a page-break with \\*
\newtheoremstyle{plain}
  {}   				% ABOVESPACE
  {}   				% BELOWSPACE
  {\itshape}  % BODYFONT
  {}       		% INDENT (empty value is the same as 0pt)
  {\mdseries\scshape} % HEADFONT
  {.}         % HEADPUNCT
  { } 				% HEADSPACE
  {\thmname{#1}\thmnumber{ #2}\ifx#3\empty\else\ (#3)\fi}
\theoremstyle{plain}
\newtheorem{theorem}{\underline{Theorem}}
\newtheorem{lemma}[theorem]{\underline{Lemma}}       	
\newtheorem{proposition}[theorem]{\underline{Proposition}}

\newtheoremstyle{definition}
  {}   				% ABOVESPACE
  {}   				% BELOWSPACE
  {}  				% BODYFONT
  {}      		% INDENT (empty value is the same as 0pt)
  {\mdseries\scshape} % HEADFONT
  {.}         % HEADPUNCT
  { } 				% HEADSPACE
  {\thmname{#1}\thmnumber{ #2}\ifx#3\empty\else\ (#3)\fi}
\theoremstyle{definition}
\newtheorem{definition}[theorem]{\underline{Definition}}
\newtheorem{example}[theorem]{\underline{Example}}
\newtheorem{remark}[theorem]{\underline{Remark}}

\renewcommand{\[}{\left[}

\newcommand\Cb{\mathds{C}}
\newcommand\Eb{\mathds{E}}
\newcommand\Fb{\mathds{F}}

\newcommand\Pb{\mathds{P}}

\newcommand\Rb{\mathds{R}}

\newcommand\Ac{\mathscr{A}}
\newcommand\Fc{\mathscr{F}}
\newcommand\Gc{\mathscr{G}}
\newcommand\Lc{\mathscr{L}}
\newcommand\Oc{\mathscr{O}}
\newcommand\Pc{\mathscr{P}}

\newcommand\Zc{\mathscr{Z}}

\newcommand\eps{\epsilon}		% change to \varepsilon if you do not use LGRgreek option in mathastext package
\newcommand\om{\omega}
\newcommand\Om{\Omega}
\newcommand\sig{\sigma}
\newcommand\Sig{\Sigma}
\newcommand\Lam{\Lambda}
\newcommand\gam{\gamma}
\newcommand\Gam{\Gamma}
\newcommand\lam{\lambda}
\newcommand\del{\delta}

\newcommand\zb{\bar{z}}
\newcommand\Tb{{\overline{T}}}

\newcommand{\vb}{\bar{v}}
\newcommand{\sigb}{\bar{\sig}}

\newcommand\Cv{\mathbf{C}} %<--- with concrete fonts there is no bolded math font so you changed to bold text font
\newcommand\mv{\textbf{m}} %<--- with concrete fonts there is no bolded math font so you changed to bold text font

\newcommand\psih{\widehat{\psi}}

\newcommand\Ebt{\widetilde{\Eb}}
\newcommand\Pbt{\widetilde{\Pb}}
\newcommand\Act{\widetilde{\Ac}}
\newcommand\Yt{\widetilde{Y}}
\newcommand\Wt{\widetilde{W}}

\renewcommand\d{\partial}

\newcommand\ii{\mathtt{i}}
\newcommand\dd{\mathrm{d}}
\newcommand\ee{\mathrm{e}}

\newcommand{\T}{\top}
\newcommand{\BS}{\text{B}}
\newcommand{\Ff}{\mathfrak{F}}
\newcommand{\Gf}{\mathfrak{G}}
\newcommand{\Hf}{\mathfrak{H}}

\begin{document}

\title{Explicit Caplet Implied Volatilities for Quadratic Term-Structure Models}

\author{
Matthew Lorig
\thanks{Department of Applied Mathematics, University of Washington.  \textbf{e-mail}: \url{mlorig@uw.edu}}
\and
Natchanon Suaysom
\thanks{Department of Applied Mathematics, University of Washington.  \textbf{e-mail}: \url{nsuaysom@uw.edu}}
}

\date{This version: \today}

\maketitle

%\tableofcontents
\begin{abstract}
We derive an explicit asymptotic approximation for implied volatilities of caplets under the assumption that the short-rate is described by a generic quadratic term-structure model.  In addition to providing an asymptotic accuracy result, we perform experiments in order to gauge the numerical accuracy of our approximation.
\end{abstract}

\\[0.5em]
\\[0.5em]

%-----------------------------------------------------------------------------------
%
%       SECTION: 		Introduction
%
%-----------------------------------------------------------------------------------

\section{Introduction}
In a general \textit{term-structure} framework, the instantaneous short-rate of interest is given by an explicit function of time and some auxiliary factors, which are typically modeled as the solution of a (multi-dimensional) stochastic differential equation.  By far the most well-known class of term-structure models are the \textit{affine term-structure} (ATS) models which, as the name suggests, model the short rate as an affine function of the auxiliary factor process.  Notable ATS models include the Vasicek \cite{vasicek1977equilibrium}, Cox-Ingersoll-Ross {(CIR)} \cite{cox2005theory} and Hull-White \cite{hull1990pricing} models.  ATS models have enjoyed widespread popularity because they allow for zero-coupon bond prices to be written explicitly as exponential affine functions of the auxiliary factors.
\\[0.5em]
A somewhat lesser-known class of term-structure models are the \textit{quadratic term-structure} (QTS) models, which, as the name suggests, model the short rate as an quadratic function of the auxiliary factor process.  
QTS models include some ATS models as special cases and also offer some additional modeling flexibility due to the fact that the zero-coupon bond price can be written as exponential quadratic functions of the auxiliary factors.  Moreover, empirical results from \cite{ahn2002quadratic} indicate that QTS better capture historical bond price than ATS models.
\\[0.5em]
The purpose of this paper is to derive an explicit approximation for implied volatilities of options written on simple forward rates, assuming the underlying short-rate is given by a general QTS model.  The implied volatility approximation we obtain is based on the coefficient polynomial expansion method that was introduced by \cite{pagliarani2011analytical} in order to derive approximate option prices in a scalar setting and later extended in \cite{lorig-pagliarani-pascucci-2} in order to derive approximate implied volatilities in a multi-factor local-stochastic volatility (LSV) setting.  Our work is similar in some senses to \cite{lorig2022options}, who also employ the polynomial expansion method to derive approximate implied volatilities.  But, there are two important differences between that work and ours: (i) \cite{lorig2022options} derived implied volatilities for options on bonds rather than on simple-forward rates and (ii) \cite{lorig2022options} focus on ATS rather than QTS models.  For related work on implied volatility in a Heath-Jarrow-Morton (HJM) setting, we refer the reader to \cite{angelini2006notes}.
\\[0.5em]
The rest of the paper proceeds as follows:  
in Section \ref{sec:model} we introduce a financial market in which the short-rate of interest is described by the class of QTS models. In Section \ref{sec:pricing} we provide a concise review of how to explicitly price options on bonds and simple forward rates (including caplets) in a QTS setting using Fourier transforms.
In Section \ref{sec:lsv-connection} we provide a precise link (see Proposition \ref{prop:connection} and Remark \ref{rmk:lsv}) between simple forward rates and classical multi-factor LSV models.  We use this result in Section \ref{sec:price-asymptotics} 
to develop an explicit approximation for caplet prices.
In Section \ref{sec:imp-vol}, we translate the price approximation into a corresponding approximation of implied volatility.
Lastly, in Section \ref{sec:examples}, {we perform experiments to gauge the numerical accuracy of our implied volatility approximation.}

%-----------------------------------------------------------------------------------
%
%       SECTION: 		Model and Assumptions
%
%-----------------------------------------------------------------------------------

\section{Quadratic term-structure models}
\label{sec:model}
We fix a time horizon $\Tb < \infty$ and consider a continuous-time financial market, defined on a filtered probability space $(\Om,\Fc,\Fb,\Pb)$ with no arbitrages and no transaction costs.  The probability measure $\Pb$ represents the market's chosen pricing measure taking the \textit{money market account} $M = (M_t)_{0 \leq t \leq \Tb}$ as num\'eraire.  The filtration $\Fb = (\Fc_t)_{0 \leq t \leq \Tb}$ represents the history of the market.
\\[0.5em]
We suppose that the money market account $M$ is strictly positive, continuous and non-decreasing.  As such, there exists a non-negative $\Fb$-adapted \textit{short-rate} process $R = (R_t)_{0 \leq t \leq \Tb}$ such that
\begin{align}
\dd M_t
	&=	R_t M_t \, \dd t ,  &
M_0
	&> 0 . \label{eq:dM}
\end{align}
We will focus on the case in which the dynamics of the short-rate $R$ are described by a QTS model.
Specifically, let $Y = (Y_t^{(1)}, Y_t^{(2)}, \ldots, Y_t^{(d)})^\top_{0\leq t \leq \Tb}$, be the unique strong solution of a stochastic differential equation (SDE) of the following form
\begin{align}
\dd Y_t
	&=	(  \lam + \Lam \, Y_t ) \, \dd t + \Sig \, \dd W_t , \label{eq:dY} 
\end{align}
where $\lam \in \Rb_+^{d \times 1}$ is a column vector, the matrix $\Lam \in \Rb^{d \times d}$ is diagonalizable and has negative real components of eigenvalues, the matrix $\Sig \in \Rb^{d \times d}$ and $W = (W_t^{(1)}, W_t^{(2)}, \ldots, W_t^{(d)})_{0 \leq t \leq \Tb}^\top$ is a $d$-dimensional $(\Pb,\Fb)$-Brownian motion. Then, following \cite[Section 2]{ahn2002quadratic}, every equivalence class of QTS models has a unique \textit{canonical representation} of the form
\begin{align}
R_t
	&=	r(Y_t)
	:=	q +  Y_t^{\T}  \Xi \, Y_t , 	\label{eq:r-def} 
\end{align}
for some constant $q \in \Rb_+$ and some matrix $\Xi \in \Rb^{d \times d}$ that is positive semidefinite and satisfies $\Xi_{i,i} = 1$ for $i = 1, 2, \ldots, d$.  Note that the restrictions on $q$ and $\Xi$ guarantee that the short-rate $R$ is non-negative.

%-----------------------------------------------------------------------------------
%
%       SECTION: 		Pricing options on bonds and forward rates
%
%-----------------------------------------------------------------------------------

\section{Pricing options on bonds and simple forward rates in a QTS setting}
\label{sec:pricing}
In this section we provide a formal review of how to explicitly price options on bonds and simple forward rates in the QTS setting. For a rigorous treatment of the results presented below, we refer the reader to \cite[Section 4.3]{chen2004quadratic}.  To begin, for any $T \leq \Tb$, column vector $\nu \in \Cb^{d \times 1}$ and matrix $\Omega \in \Cb^{d\times d}$, let us define $\Gam(\,\cdot\,,\,\cdot\,;T,\nu,\Omega) : [0,T] \times \Rb^{d \times 1} \to \Cb$ by
\begin{align}
\Gam(t,Y_t;T,\nu, \Omega) 
	&:= \Eb_t \exp \Big( -\int_t^T \dd s \,r(Y_s) + \nu^\top Y_T  + Y_T^\top  \Omega \, Y_T   \Big) , \label{eq:Gamma-def} 
\end{align}
where $\Eb_t$ denotes the $\Fc_t$-conditional expectation under $\Pb$. The existence of the function $\Gam$ follows from the Markov property of $Y$.  Formally, the function $\Gam$ satisfies the Kolmogorov backward partial differential equation (PDE)
\begin{align}
(\d_t + \Ac - r ) \Gam(t,\, \cdot \, ; T,\nu, \Omega)
  &=  0 , &
\Gam(T,y;T,\nu,\Omega)
  &=  \exp \Big( \nu^{\T} y  + y^{\T} \Omega \, y \Big) , \label{eq:Gamma-pde} 
	%&=	\red{ \exp \Big( y^{\T} \nu  + y^{\T} \Omega y \Big) } , &\text{(OLD)}
\end{align}
where the operator $\Ac$ is the generator of $Y$ under $\Pb$.  Explicitly, the generator $\Ac$ is given by
\begin{align}
\Ac
	&=	( \lam + \Lam y )^\top \nabla_y + \tfrac{1}{2} \text{Tr}(\Sig \Sig^\top \nabla_y \nabla_y^\top ) , \label{eq:A} 
\end{align}
where $\nabla_y = (\d_{y_1}, \d_{y_2}, \ldots, \d_{y_d})^\top$, and ``$\text{Tr}$'' denotes the trace operator.
The solution to \eqref{eq:Gamma-pde} is given by
\begin{align}
\Gam(t,y;T,\nu,\Omega) 
	&=	\exp \Big( - F(t;T,\nu,\Omega) -  G^\top(t;T,\nu,\Omega) y - y^{\T} H(t;T,\nu,\Omega)y  \Big) , \label{eq:Gamma-explicit} 
\end{align}
where, from \cite[Theorem 3.6]{chen2004quadratic}, the scalar-valued function $F(\, \cdot \, ;T,\nu,\Om) : [0,T] \to \Cb$, the vector-valued function $G(\, \cdot \, ;T,\nu,\Om) : [0,T] \to \Cb^{d \times 1}$ and the matrix-valued function $H(\, \cdot \, ;T,\nu,\Om) : [0,T] \to \Cb^{d \times d}$ solve the following system of ordinary differential equations (ODEs) 
\begin{align}
\d_t F(t;T,\nu,\Omega) & = \tfrac{1}{2}G^{\T}(t;T,\nu,\Omega)\Sig \Sig ^{\T}G(t;T,\nu,\Omega)-\text{Tr}(\Sig \Sig^{\T}H(t;T,\nu,\Omega))-G^{\T}(t;T,\nu,\Omega)\lam-q,
\\
F(T;T,\nu,\Omega) & = 0, \label{eq:Fode}
\\
\d_t G(t;T,\nu,\Omega) & = 2H^\T(t;T,\nu,\Omega)\Sig \Sig^\T G(t;T,\nu,\Omega)- \Lam G(t;T,\nu,\Omega)-2H^\T(t;T,\nu,\Omega) \lam ,
\\
G(T;T,\nu,\Omega) & = -\nu, \label{eq:Gode}
\\
\d_t H(t;T,\nu,\Omega) & =  2H^\T(t;T,\nu,\Omega)\Sig\Sig^\T H(t;T,\nu,\Omega)-\Lam H(t;T,\nu,\Omega)- H^\T(t;T,\nu,\Omega) \Lam^\T - \Xi,
\\
H(T;T,\nu,\Omega) & = -\Omega, \label{eq:Hode}
\end{align}
The solution to the system \eqref{eq:Fode},\eqref{eq:Gode} and \eqref{eq:Hode} exists and is unique.
\\[0.5em]
As $F(t;T,0,0)$, $G(t;T,0,0)$ and $H(t;T,0,0)$ will appear frequently throughout this paper, it will be convenient to define
\begin{align}
\Ff(t;T)
	&:=F(t;T,0,0) , &
\Gf(t;T)
	&:=G(t;T,0,0) , &
\Hf(t;T)
	&:=H(t;T,0,0) . \label{eq:F0-G0-H0}
\end{align}
%\\[0.5em]
Now, for any $T \leq \Tb$, let us denote by $B^T = (B_t^T)_{0 \leq t \leq T}$ the value of a \textit{zero-coupon bond} that pays one unit of currency at time $T$.
In the absence of arbitrage, the process $B^T/M$ must be a $(\Pb,\Fb)$-martingale.  As such, we have
\begin{align}
\frac{B_t^T}{M_t} 
  &= \Eb_t \Big( \frac{B_T^T}{M_T} \Big) = \Eb_t \Big( \frac{1}{M_T} \Big) ,
\end{align}
where we have used $B_T^T = 1$.  
Solving for $B_t^T$, we obtain
\begin{align}
B_t^T
  &=  \Eb_t \Big( \frac{M_t}{M_T} \Big) = \Eb_t \Big( \ee^{- \int_t^T \dd s \, r(Y_s)} \Big) = \Gam(t,Y_t;T,0,0) \\
	&=	\exp \Big(-\Ff(t;T) - \Gf^\T(t;T) Y_t  - Y^{\T}_t \Hf(t;T)  Y_t \Big) , \label{eq:B-explicit} 
\end{align}
where the third equality follows from \eqref{eq:Gamma-def} and the fourth equality follows from \eqref{eq:Gamma-explicit} and \eqref{eq:F0-G0-H0}. 

\subsection{Pricing options on zero-coupon bonds}
Let $U = (U_t)_{0 \leq t \leq T}$ denote the value of a European option that pays $\psi( \log B_T^\Tb )$ at time $T$ for some function $\psi : \Rb_- \to \Rb$.  With the aim of finding $U_t$, let $\psih:\Cb \to \Cb$ denote the generalized Fourier transform of $\psi$, which {is defined} as follows
\begin{align}
\psih(\om)
	&:=	\int_{-\infty}^{\infty} \dd x \, \ee^{-\ii \om x} \psi(x) , &
\om
	&=	\om_r + \ii \om_i , &
\om_r, \om_i
	&\in \Rb . \label{eq:ft}
\end{align}
We can recover $\psi$ from $\psih$ using the inverse Fourier transform
\begin{align}
\psi(x)
	&=	\frac{1}{2\pi} \int_{-\infty}^{\infty} \dd \om_r \, \ee^{\ii \om x} \psih(\om) . \label{eq:ift}
\end{align}
Now, noting that, in the absence of arbitrage, the process $U/M$ must be a $(\Pb,\Fb)$-martingale, we have
\begin{align}
\frac{U_t}{M_t} 
	&= \Eb_t \Big( \frac{U_T}{M_T} \Big) 
	= 	\Eb_t \Big( \frac{\psi(\log B_T^\Tb)}{M_T} \Big) .
\end{align}
Solving for $U_t$, we obtain
\begin{align}
U_t
	&=	\Eb_t \exp \Big( -\int_t^T \dd s \, r(Y_s)  \Big) \psi( \log B_T^\Tb ) \label{eq:conditional-expectation} \\
	&=	\frac{1}{2\pi} \int_{-\infty}^{\infty} \dd \om_r \, \psih(\om) \Eb_t \ee^{  -\int_t^T \dd s \, r(Y_s)  + \ii \om \log B_T^\Tb } \\
	&=	\frac{1}{2\pi} \int_{-\infty}^{\infty} \dd \om_r \, \psih(\om) \Eb_t \ee^{ -\int_t^T \dd s \, r(Y_s)  } \Eb_T \ee^{ \ii \om \log B_T^\Tb } \\
	&=	\frac{1}{2\pi} \int_{-\infty}^{\infty} \dd \om_r \, \psih(\om) \ee^{ - \ii \om \Ff(T;\Tb) } %\\	& \quad 
			\Eb_t \ee^{ -\int_t^T \dd s \, r(Y_s)  - \ii \om  \Gf^\T(T;\Tb) Y_T - \ii \om Y^\T_T \Hf(T;\Tb) Y_T } \\
	&=	\frac{1}{2\pi} \int_{-\infty}^{\infty} \dd \om_r \, \psih(\om) \ee^{ - \ii \om \Ff(T;\Tb) } 
			\Gam\big(t,Y_t;T,-\ii \om \Gf(T;\Tb), -\ii \om \Hf(T;\Tb) \big)	\label{eq:u-integral} \\
	&=: u(t,Y_t;T,\Tb) . \label{eq:u-def}
\end{align}
In general, the inverse Fourier integral \eqref{eq:u-integral} that defines $u$ must be computed numerically.  

\begin{remark}
As $\log B_T^\Tb \leq 0$ $\Pb$-a.s., values of $\psi(x)$ for $x > 0$ do not affect the conditional expectation \eqref{eq:conditional-expectation} and thus do not affect the value $U_t$ of the option.  The values of $\psi(x)$ for $x > 0$ do, however, affect convergence properties of the Fourier transform \eqref{eq:ft} and inverse Fourier transform  \eqref{eq:ift}.  As such, it makes sense to choose values of $\psi(x)$ for $x>0$ so that these integrals converge for some value of $\om_i \in \Rb$.
\end{remark}

\subsection{Pricing options on simple forward rates}
\label{sec:options-on-forward-rates}
The \emph{simple forward rate from $T$ to $\Tb$} is a process $L^{T,\Tb} = (L_t^{T,\Tb})_{0 \leq t \leq T}$, which is defined as follows
\begin{align}
 L^{T,\Tb}_{t} 
	& := \frac{1}{\tau }\Big(\frac{B^{T}_t}{B^{\Tb}_t}-1 \Big), &
\tau
	&:=	\Tb - T . \label{eq:L-def}
\end{align}
Let $V = (V_t)_{0 \leq t \leq T}$ denotes the value of a European forward rate option with \emph{reset date} $T$ and \emph{settlement date} $\Tb$ that pays $\phi(\log L^{T,\Tb}_T)$ at time $\Tb$ for some function $\phi : \Rb \to \Rb$. Because the payoff $\phi(\log L^{T,\Tb}_T)$ to be made at time $\Tb$ is known at time $T$ we have
\begin{align}
V_T
	&=	B_T^\Tb \phi( \log L_T^{T,\Tb} ) . \label{eq:V=B-phi}
\end{align}
To see this, simply note that $V_\Tb = B_\Tb^\Tb \phi( \log L_T^{T,\Tb} ) = \phi( \log L_T^{T,\Tb} )$.
Using \eqref{eq:L-def} and $B_T^T = 1$, we can express $V_T$ as a function of $B_T^\Tb$ as follows
\begin{align}
V_T
	&=	B_T^\Tb \phi \Big( \log \Big[  \frac{1}{\tau} \Big(\frac{1}{B_T^\Tb}-1 \Big) \Big] \Big) 
	=	\ee^{ \log B_T^\Tb } \phi \Big( \log \Big[  \frac{1}{\tau} \Big( \ee^{ - \log B_T^\Tb}-1 \Big) \Big] \Big)
	=: \psi( \log B_T^\Tb ) . \label{eq:psi-def}
\end{align}
Thus, we can view a forward rate option on $L^{T,\Tb}$ with reset date $T$, settlement date $\Tb$ and payoff $\phi(\log L_T^{T,\Tb})$ as a European option on $B^T$ with expiration date $T$ and payoff $\psi( \log B_T^\Tb )$, where $\psi$ is defined in \eqref{eq:psi-def}.
It follows that
\begin{align}
V_t
	&=	u(t,Y_t;T,\Tb) , &
	&\text{where}&
\psi(x)
	&=	\ee^{ x } \phi \Big( \log \Big[  \frac{1}{\tau} \Big( \ee^{ - x}-1 \Big) \Big] \Big) , \label{eq:V-explicit}
\end{align}
with $u$ given by \eqref{eq:u-def}.

\begin{example}
An important example of a European forward rate option is a \textit{caplet}, which has a payoff
\begin{align}
\phi( \log L_T^{T,\Tb} )
	&=	\tau ( \ee^{\log L_T^{T,\Tb}} - \ee^{k} )^+ .
\end{align}
Here, $k := \log K$ is the $\log$ strike of the caplet.  We have from \eqref{eq:V-explicit} that
\begin{align}
\psi( x )
	&=	( 1 + \tau \ee^k) \Big( \frac{1}{1 + \tau \ee^k} - \ee^{ x } \Big)^+ ,
\end{align}
and thus, from \eqref{eq:ft}, the generalized Fourier transform of $\psi$ is given by
\begin{align}
\psih(\om)
	%&=	- ( 1 + \tau \ee^k) \frac{\ee^{\ell- \ii \ell \om }}{\om^2 + \ii\om } 
	%=		- ( 1 + \tau \ee^k) \frac{ (1 + \tau \ee^k )^{\ii \om - 1} }{\om^2 + \ii\om } 
	&=		\frac{ -  \, (1 + \tau \ee^k )^{\ii \om} }{\om^2 + \ii\om } , &
%\ell
	%&=	\log \frac{1}{1 + \tau \ee^k} .
\om_i 
	&> 0 . \label{eq:caplet-ft}
\end{align}
The caplet price $V_t$ can now be computed by inserting the expression \eqref{eq:caplet-ft} for $\psih$ into \eqref{eq:u-integral} and evaluating the integral numerically.
\end{example}

%-----------------------------------------------------------------------------------
%
%       SECTION: 		Relation to local-stochastic volatility models
%
%-----------------------------------------------------------------------------------

\section{Relation between QTS models and LSV models}
\label{sec:lsv-connection}
While \eqref{eq:u-def} and \eqref{eq:V-explicit} in conjunction with \eqref{eq:caplet-ft} can be used to compute caplet prices explicitly, the resulting expression tells us very little about the corresponding implied volatilities.  In this section, we will establish a precise relation between QTS models and LSV models.  This relation will be used in subsequent sections to find an explicit approximation for caplet prices and implied volatilities.
\\[0.5em]
We begin by deriving the dynamics of $B^T/M$.  Using \eqref{eq:dM} and \eqref{eq:B-explicit}, we have by It\^o's Lemma that 
\begin{align}
\dd \Big( \frac{B_t^T}{M_t} \Big)
	&= 	 \Big( \frac{B_t^T}{M_t} \Big) \gam^\top(t,Y_t;T) \dd W_t , \label{eq:BoverM} 
\end{align}
where we have introduced the vector-valued function $\gam(\, \cdot \, , \, \cdot \, ;T) : [0,T] \times \Rb^{d \times 1} \to \Rb^{d \times 1}$, which is defined as follows
\begin{align}
\gam(t,Y_t;T) 
	&:=	\Sig^\top \nabla_y \log \Gamma(t,Y_t;T,0,0) \\
	&=	- \Sig^\top \Big( \Gf(t;T) + \Big( \Hf(t;T) + \Hf^\T(t;T) \Big) Y_t \Big). \label{eq:gamma-def} 
\end{align}
%\\[0.5em]
Now, let us denote by $\Pbt$ the \textit{$\Tb$-forward probability measure}, whose relation to $\Pb$ is given by the following Radon-Nikodym derivative
\begin{align}
\frac{\dd \Pbt}{\dd \Pb} 
  &:= \frac{M_0 B_\Tb^\Tb}{B_0^\Tb M_\Tb} 
	= \exp \Big( - \frac{1}{2} \int_0^\Tb \| \gam(t,Y_t;\Tb) \|^2 \dd t + \int_0^\Tb \gam^\top(t,Y_t;\Tb) \dd W_t \Big) , \label{eq:measure-change} 
\end{align}
where $\| \gam \|^2 = \gam^\top \gam$.  Observe that the last equality follows from \eqref{eq:BoverM}.  By Girsanov's theorem and \eqref{eq:measure-change}, the process $\Wt = (\Wt_t^{(1)}, \Wt_t^{(2)}, \ldots, \Wt_t^{(d)})_{0 \leq t \leq \Tb}^\top$, defined as follows
\begin{align}
 \Wt_t 
	&:=  - \int_0^t \gam(s,Y_s;\Tb) \dd s + W_t , \label{eq:W_tilde} 
\end{align}
is a $d$-dimensional $(\Pbt,\Fb)$-Brownian motion.  The following lemma will be useful.
\begin{lemma}
\label{lemma:forward-price}
Let $\Pi = (\Pi_t)_{0 \leq t \leq \Tb}$ denotes the value of a self-financing portfolio 
and let $\Pi^\Tb = (\Pi_t^\Tb)_{0 \leq t \leq \Tb}$, 
defined by $\Pi_t^\Tb := \Pi_t/B_t^\Tb$, be the \textit{$\Tb$-forward price of $\Pi$}. Then the process $\Pi^\Tb$ is a $(\Pbt,\Fb)$-martingale.
\end{lemma}
\begin{proof}
Define the \textit{Radon-Nikodym derivative process} $Z = (Z_t)_{0 \leq t \leq \Tb}$ by $Z_t := \Eb_t (\dd \Pbt/ \dd \Pb)$.  Using the fact that $\Pi/M$ is a $(\Pb,\Fb)$-martingale as well as \cite[Lemma 5.2.2]{shreve2004stochastic} we have for any $0 \leq t \leq T \leq \Tb$ that
\begin{align}
\frac{\Pi_t}{M_t}
  &=  \Eb_t \Big( \frac{\Pi_T}{M_T} \Big)
  =   Z_t \Ebt_t \Big( \frac{1}{Z_T} \frac{\Pi_T}{M_T} \Big)
  =   \frac{B_t^\Tb}{M_t} \Ebt_t \Big( \frac{M_T}{B_T^\Tb} \frac{\Pi_T}{M_T} \Big) , \label{eq:1}
\end{align}
where $\Ebt_t$ denotes the $\Fc_t$-conditional expectation under $\Pbt$.  Dividing both sides of equation \eqref{eq:1} by $B_t^\Tb$ and canceling common factors of $M_t$ and $M_T$, we obtain
\begin{align}
\Pi_t^\Tb 
  &=  \frac{\Pi_t}{B_t^\Tb}
   =  \Ebt_t \frac{\Pi_T}{B_T^\Tb}
   =  \Ebt_t  \Pi_T^\Tb  ,
\end{align}
which establishes that $\Pi^\Tb$ is a $(\Pbt,\Fb)$-martingale, as claimed.
\end{proof}

\noindent
Note from \eqref{eq:L-def} that $L^{T,\Tb}$ is the $\Tb$-forward price of a static portfolio consisting of $1/\tau$ shares of $B^T$ and $-1/\tau$ shares of $B^\Tb$.  As such, we have from Lemma \ref{lemma:forward-price} that $L^{T,\Tb}$ is a $(\Pbt,\Fb)$-martingale. 
\\[0.5em] 
It will be helpful to write the dynamics of $L^{T,\Tb}$ under the $\Tb$-forward measure $\Pbt$.
Using It\^o's Lemma, \eqref{eq:L-def}, \eqref{eq:BoverM} and \eqref{eq:W_tilde}, we obtain
\begin{align}
\dd L^{T,\Tb}_{t} 
		&=  \frac{1}{\tau } \dd \Big( \frac{B^{T}_t}{B^{\Tb}_t} \Big) 
		=		 \frac{1}{\tau } \dd \Big( \frac{B^{T}_t/M_t}{B^{\Tb}_t/M_t} \Big) \\
		&=	\frac{1}{\tau }  \Big( \frac{B^{T}_t}{B^{\Tb}_t} \Big) \Big( \gam^\top(t,Y_t;T) - \gam^\top(t,Y_t;\Tb) \Big) \dd \Wt_t \\
		&=	 \Big( L^{T,\Tb}_{t}+\frac{1}{\tau } \Big) \Big( \gam^\top(t,Y_t;T) - \gam^\top(t,Y_t;\Tb) \Big) \dd \Wt_t . \label{eq:L-dynamics} 
\end{align} 
Now, let us denote by $X = (X_t)_{0 \leq t \leq T}$ the $\log$ of the simple forward rate from $T$ to $\Tb$, that is
\begin{align}
X_t
  & :=  \log L_t^{T,\Tb} . \label{eq:X-def}
\end{align}
We are now in a position to state the main result of this section.

\begin{proposition}
\label{prop:connection}
As in Section \ref{sec:options-on-forward-rates}, let $V = (V_t)_{0 \leq t \leq T}$ denote the value of a European forward rate option with {reset date} $T$ and {settlement date} $\Tb$ that pays $\phi(\log L^{T,\Tb}_T) = \phi(X_T)$ at time $\Tb$ for some function $\phi : \Rb \to \Rb$. 
Let $V^\Tb = (V_t^\Tb)_{0 \leq t \leq T}$ denote the $\Tb$-forward price of $V$.  Then, there exists a function 
$v(\, \cdot \, ,\, \cdot \, ,\, \cdot \, ;T,\Tb) : [0,T] \times \Rb^{} \times \Rb^{d \times 1} \to \Rb$ such that
\begin{align}
V_t^\Tb
	=		v(t,X_t,Y_t;T,\Tb) .
\end{align}
Moreover, the function $v$ satisfies the following PDE
\begin{align}
( \d_t + \Act(t) ) v(t,\cdot,\cdot;T,\Tb)
	&=	0 , &
v(T,x,y;T,\Tb)
	&=	\phi(x) ,
	\label{eq:v-pde}
\end{align}
where the operator $\Act$ is given by
\begin{align}
\Act(t)
	&=	\frac{1}{2} \Big( 1 + \frac{\ee^{-x}}{\tau} \Big)^2 \| \gam(t,y;T) - \gam(t,y;\Tb) \|^2 (\d_x^2 - \d_x) \\ & \quad
			+ \Big( \lam + \Lam y  + \Sig \, \gam(t,y;\Tb) \Big)^\top \nabla_y + \frac{1}{2} \textup{Tr}(\Sig \Sig^\top \nabla_y \nabla_y^\top ) \\ &\quad
			+ \Big( 1 + \frac{\ee^{-x}}{\tau} \Big) \Big( \Sig \, \gam(t,y;T) - \Sig \, \gam(t,y;\Tb) \Big)^\top \nabla_y \d_x . \label{eq:A-tilde}
\end{align}
\end{proposition}

\begin{proof}
We begin by writing the dynamics of $X$ and $Y$ under the $\Tb$-forward probability measure $\Pbt$.
First, using It\^o's Lemma and \eqref{eq:L-dynamics}, we obtain
\begin{align}
\dd X_t
	&=	- \frac{1}{2} \Big( 1 + \frac{\ee^{-X_t}}{\tau} \Big)^2 \| \gam(t,Y_t;T) - \gam(t,Y_t;\Tb) \|^2 \dd t \\ &\quad
			+ \Big( 1 + \frac{\ee^{-X_t}}{\tau} \Big) \Big( \gam^\top(t,Y_t;T) - \gam^\top(t,Y_t;\Tb) \Big) \dd \Wt_t .
\end{align}
Next, using \eqref{eq:dY} and \eqref{eq:W_tilde}, we find
\begin{align}
\dd Y_t
	&=	\Big( \lam + \Lam \, Y_t  + \Sig \, \gam(t,Y_t;\Tb) \Big) \dd t + \Sig \, \dd \Wt_t . \label{eq:dY-forward}
\end{align}
The pair $(X,Y)$ is a Markov process whose generator $\Act$ under $\Pbt$ is given by \eqref{eq:A-tilde}.
Now, using the fact that $\Tb$-forward prices are $(\Pbt,\Fb)$-martingales and the fact that the process $(X,Y)$ is Markov, there exists a function $v$ such that
\begin{align}
V_t^\Tb
	&= 	\frac{V_t}{B_t^\Tb}
	=		\Ebt_t \frac{V_T}{B_T^\Tb}
	=		\Ebt_t \phi(X_T) 
	=		v(t,X_t,Y_t;T,\Tb) , \label{eq:v-def} 
\end{align}
where, in the third equality, we have used \eqref{eq:V=B-phi}.  We have from \eqref{eq:v-def} that the function $v$ satisfies the Kolmogorov backward PDE \eqref{eq:v-pde}.
\end{proof}

\noindent
A few important remarks are in order.

\begin{remark}
\label{rmk:lsv}
As $L^{T,\Tb} = \ee^X$ is a positive $(\Pbt,\Fb)$-martingale, the process $(X,Y)$ can be seen as a classical LSV model, where $Y$ represents non-local factors of volatility.  For example, when $d=1$ we have from \eqref{eq:A-tilde} that
\begin{align}
\Act(t)	
	&=	c(t,x,y) (\d_x^2 - \d_x ) + f(t,x,y) \d_y + g(t,x,y) \d_y^2  + h(t,x,y) \d_x \d_y, \label{eq:A-1d}
\end{align}
where the functions $c$, $f$, $g$ and $h$ are given by 
\begin{align}
c(t,x,y)
	&= \tfrac{1}{2} \Sig^2 \Big( 1 + \frac{\ee^{-x}}{\tau} \Big)^2 \Big( \Gf(t;\Tb) - \Gf(t;T) + 2 \Big( \Hf(t;\Tb) - \Hf(t;T) \Big) y \Big)^2, \\
f(t,x,y)
	& = \lam + \Lam y - \Sig^2\Big( \Gf(t;\Tb) +2 \Hf(t;\Tb) y \Big), \\
g(t,x,y)
	&=	\tfrac{1}{2} \Sig^2 , \\
h(t,x,y)
	&=	\Sig^2 \Big( 1 + \frac{\ee^{-x}}{\tau} \Big) \Big( \Gf(t;\Tb) - \Gf(t;T) + 2 \Big( \Hf(t;\Tb) - \Hf(t;T) \Big) y \Big) .
\end{align}
\end{remark}

\begin{remark}
\label{remark:elliptic}
The instantaneous covariance matrix of the process $(X,Y)$ is singular due to the fact that $X_t$ can be written as an explicit function of $t$ and $Y_t$.  Indeed, using \eqref{eq:B-explicit}, \eqref{eq:L-def} and \eqref{eq:X-def}, we have
\begin{align}
X_t
	&=	\log \Big[ \frac{1}{\tau} \Big( \ee^{ \Ff(t;\Tb) - \Ff(t;T) +  \Big( \Gf^\T(t;\Tb) - \Gf^\T(t;T) \Big) Y_t 
			+ Y^{\T}_t \Big( \Hf(t;\Tb) - \Hf(t;T) \Big) Y_t} - 1 \Big) \Big] 
	=:	\xi(t,Y_t;T,\Tb) . \label{eq:xi-def}
\end{align}
As a result, the generator $\Act$ defined in \eqref{eq:A-tilde} is not uniformly elliptic.  The setting here is similar to the settings in \cite{LETF} and \cite{VIX} where the authors use the approximation methods described in Sections \ref{sec:price-asymptotics} and \ref{sec:imp-vol} of this paper to find explicit approximations of implied volatility for options on leveraged exchange traded funds and the VIX, respectively.  As the authors of those papers point out, the lack of a uniformly elliptic generator does \textit{not} complicate the construction of a formal implied volatility approximation.
\end{remark}

\begin{remark}
Let $\Yt^{(-j)} := (Y_t^{(1)}, \ldots, Y_t^{(j-1)},Y_t^{(j+1)}, \ldots, Y_t^{(d)})_{0 \leq t \leq T}^\top$.  Had the function $\xi(t,y;T,\Tb)$ defined in \eqref{eq:xi-def} been invertible with respect to $y_j$ for some $j \in \{1,2,\ldots,d\}$, we could have written $Y_t^{(j)} = \xi_j^{-1}(t,X_t,\Yt_t^{(-j)};T,\Tb)$ where $\xi_j^{-1}$ is the inverse of $\xi$ with respect to $y_j$.  The process $(X,\Yt^{(-j)})$ would have been a $d$-dimensional Markov process with a non-singular instantaneous covariance matrix and thus, a uniformly elliptic generator.  This was the approach taken in \cite{lorig2022options}, where the authors found explicit approximations of implied volatilities for options on bonds in an \textit{affine} (as opposed to \textit{quadratic}) term-structure setting.
\end{remark}

\begin{remark}
Clearly, because $V_t^\Tb = V_t/B_t^\Tb$, we have from \eqref{eq:B-explicit} and \eqref{eq:V-explicit} that
\begin{align}
v(t,\xi(t,y;T;\Tb),y;T,\Tb)
	&=	\frac{u(t,y;T,\Tb)}{\Gam(t,y;\Tb,0,0)} . \label{eq:v-explicit}
\end{align}
However, as mentioned previously, the explicit expression \eqref{eq:v-explicit} for $v$ does not tell us anything about implied volatilities of caplets, which is the aim of this paper.
\end{remark}

%-----------------------------------------------------------------------------------
%
%       SECTION: 		Option price asymptotics
%
%-----------------------------------------------------------------------------------

\section{Option price asymptotics}
\label{sec:price-asymptotics}
Let $z := (x, y_1, \ldots, y_d)$.  We have from \eqref{eq:v-pde} that the function $v$ satisfies a parabolic PDE of the form
\begin{align}
( \d_t + \Act(t) ) v(t, \,\cdot \,)
	&=	0 , &
\Act(t)
	&=	\sum_{|\alpha| \leq 2} a_\alpha(t,z) \d_z^\alpha , &
v(T, \, \cdot \,)
	&=	\phi , \label{eq:pde-form}
\end{align}
where, for brevity, we have omitted the dependence on $T$ and $\Tb$.  Note that we have introduced standard multi-index notation
\begin{align}
\alpha
	&=	(\alpha_1, \alpha_2, \dots, \alpha_{d+1}) , &
\d_z^\alpha
	&=	\prod_{i=1}^{d+1} \d_{z_i}^{\alpha_i} , &
z^\alpha
	&=	\prod_{i=1}^{d+1} {z_i}^{\alpha_i} , &
| \alpha |
	&=	\sum_{i=1}^{d+1} \alpha_i , &
\alpha!
	&=	\prod_{i=1}^{d+1} \alpha_i! .
\end{align}
In general, there is no explicit solution to PDEs of the form \eqref{eq:pde-form}.  In this section, we will show in a formal manner how an explicit approximation of $v$ can be obtained by using a simple Taylor series expansion of the coefficients $(a_\alpha)_{|\alpha|\leq 2}$ of $\Act$.  The method described below was introduced for scalar diffusions in \cite{pagliarani2011analytical} and subsequently extended to $d$-dimensional diffusions in \cite{lorig-pagliarani-pascucci-2,lorig-pagliarani-pascucci-4}.  
\\[0.5em]
To begin, for any $\eps \in [0,1]$ and $\zb: [0,T] \to \Rb^{d+1}$, let $v^\eps$ be the unique classical solution to
\begin{align}
0
	&=	( \d_t + \Act^\eps(t) ) v^\eps(t, \,\cdot \,) , &
v^\eps(T, \, \cdot \,)
	&=	\phi , \label{eq:v-eps-pde} 
\end{align}
where the operator $\Act^\eps$ is defined as follows
\begin{align}
\Act^\eps(t)
	&:=	\sum_{|\alpha| \leq 2}  a_\alpha^\eps(t,z) \d_z^\alpha , &
	&\text{with}&
a_\alpha^\eps
	&:=	a_\alpha(t,\zb(t) + \eps(z - \zb(t))) , \label{eq:A-eps}
\end{align}
Observe that $\Act^\eps |_{\eps = 1} = \Act$ and thus $v^\eps |_{\eps=1} = v$.  We will seek an approximate solution of \eqref{eq:v-eps-pde} by expanding $v^\eps$ and $\Act^\eps$ in powers of $\eps$.  Our approximation for $v$ will be obtained by setting $\eps = 1$ in our approximation for $v^\eps$.   We have
\begin{align}
v^\eps
	&=	\sum_{n=0}^\infty \eps^n v_n , &
\Act^\eps(t)
	&=	\sum_{n=0}^\infty \eps^n \Act_n(t) , \label{eq:expansion}
\end{align}
where the functions $(v_n)_{n \geq 0}$ are, at the moment, unknown, and the operators $(\Act_n)_{n \geq 0}$ are given by
\begin{align}
\Act_n(t)
	&=	\frac{\dd^n }{\dd \eps^n} \Act^\eps |_{\eps=0}
	=		\sum_{|\alpha| \leq 2} a_{\alpha,n}(t,z) \d_z^\alpha , &
a_{\alpha,n}
	=		\sum_{|\eta|=n} \frac{1}{\eta!} (z - \zb(t))^\eta \d_z^\eta a_\alpha(t,\zb(t))  . \label{eq:an-taylor}
\end{align}
Note that $a_{\alpha,n}(t, \, \cdot \,)$ is the sum of the $n$th order terms in the Taylor series expansion of $a_\alpha(t,\,\cdot\,)$ about the point $\zb(t)$.  Inserting the expansions from \eqref{eq:expansion} for $v^\eps$ and $\Act^\eps$ into PDE \eqref{eq:v-eps-pde} and collecting terms of like order in $\eps$ we obtain
\begin{align}
&\Oc(\eps^0):&
0
	&=	( \d_t + \Act_0(t) ) v_0(t, \,\cdot \,) , &
v_0(T, \, \cdot \,)
	&=	\phi  , \label{eq:v0-pde} \\
&\Oc(\eps^n):&
0
	&=	( \d_t + \Act_0(t) ) v_n(t, \,\cdot \,)  + \sum_{k=1}^n \Act_k(t) v_{n-k}(t,\,\cdot\,) , &
v_n(T, \, \cdot \,)
	&=	0  . \label{eq:vn-pde}
\end{align}
Now, observe that the coefficients $(a_{\alpha,0})_{|\alpha|\leq 2}$ of $\Act_0$ do not depend on $z$.  Thus, $\Act_0$ is the generator of a $(d+1)$-dimensional Brownian motion with a time-dependent drift vector and covariance matrix.  As such, $v_0$ is given by
\begin{align}
v_0(t,z)
	&=	\Pc_0(t,T)\phi(z)
	=		\int_{\Rb^{d+1}} \dd z' \, p_0(t,z;T,z') \phi(z') . \label{eq:v0-explicit}
\end{align}
where $\Pc_0$ is the semigroup generated by $\Act_0$ and $p_0$ is the associated transition density (i.e., the solution to \eqref{eq:v0-pde} with $\phi = \del_{\zeta}$).  Explicitly, we have
\begin{align}
p_0(t,z;T,\zeta)
	&=	\frac{1}{\sqrt{(2\pi)^{d+1}|\Cv(t,T)|}}
			{\exp\left(-\tfrac{1}{2} (\zeta-z-\mv(t,T))^\top \Cv^{-1}(t,T) (\zeta-z-\mv(t,T)) \right)} , \label{eq:p0}
\end{align}
where $\mv$ and $\Cv$ are given by
\begin{align}
\mv(t,T)
		&:=	\int_t^T \dd s \, m(s) , &
\Cv(t,T)
		&:=	\int_t^T \dd s \, A(s) , \label{eq:m-and-C}
\end{align}
and $m$ and $A$ are, respectively, the instantaneous drift vector and covariance matrices
\begin{align}
m(s)
		&:=	\begin{pmatrix}
				a_{(1,0,\cdots,0),0}(s) \\ a_{(0,1,\cdots,0),0}(s) \\ \vdots \\  a_{(0,0,\cdots,1),0}(s)
				\end{pmatrix} , &
A(s)
		&:= \begin{pmatrix}
				2a_{(2,0,\cdots,0),0}(s) & a_{(1,1,\cdots,0),0}(s) & \ldots &  {a_{(1,0,\cdots,1),0}(s)} \\
				a_{(1,1,\cdots,0),0}(s) & 2a_{(0,2,\cdots,0),0}(s) & \ldots &  a_{(0,1,\cdots,1),0}(s) \\
				\vdots & \vdots & \ddots & \vdots \\
				a_{(1,0,\cdots,1),0}(s) & a_{(0,1,\cdots,1),0}(s) & \ldots &  2 a_{(0,0,\cdots,2),0}(s) \\
				\end{pmatrix} .
\end{align}
By Duhamel's principle, the solution $v_n$ of \eqref{eq:vn-pde} is
\begin{align}
v_n(t,z)
	&=	\sum_{k=1}^n \int_t^T \dd t_1 \, \Pc_0(t,t_1) \Act_k(t_1) v_{n-k}(t_1,z) \\
	&=	\sum_{k=1}^n \sum_{i \in I_{n,k}}
      \int_{t}^T \dd t_1 \int_{t_1}^T \dd t_2 \cdots \int_{t_{k-1}}^T \dd t_k \\ & \qquad 
       \Pc_0(t,t_1) \Ac_{i_1}(t_1)
       \Pc_0(t_1,t_2) \Ac_{i_2}(t_2) \cdots
       \Pc_0(t_{k-1},t_k) \Ac_{i_k}(t_k)
       \Pc_0(t_k,T)\phi(z) , \label{eq:vn-explicit} \\
I_{n,k}
    &= \{ i = (i_1, i_2, \cdots , i_k ) \in \mathds{N}^k : i_1 + i_2 + \cdots + i_k = n \} .
            \label{eq:Ink}
\end{align}
While the expression \eqref{eq:vn-explicit} for $v_n$ is explicit, it is not easy to compute as written because operating on a function with $\Pc_0$ requires performing a $(d+1)$-dimensional integral.  The following proposition establishes that $v_n$ can be expressed as a differential operator acting on $v_0$.

\begin{proposition}
\label{thm:vn}
The solution $v_n$ of PDE \eqref{eq:vn-pde} is given by
\begin{align}
v_n(t,z)
    &=  \Lc_n(t,T) v_0(t,z) , \label{eq:un} 
\end{align}
where $\Lc$ is a linear differential operator, which is given by
\begin{align}
\Lc_n(t,T)
    &=  \sum_{k=1}^n \sum_{i \in I_{n,k}}
        \int_{t}^T \dd t_1 \int_{t_1}^T \dd t_2 \cdots \int_{t_{k-1}}^T \dd t_k
        \Gc_{i_1}(t,t_1)
        \Gc_{i_2}(t,t_2) \cdots
         \Gc_{i_k}(t,t_k) , \label{eq:Ln}
\end{align}
the index set $I_{n,k}$ is as defined in \eqref{eq:Ink} and the operator $\Gc_i$ is given by
\begin{align}
\Gc_i(t,t_k) 
	&:= 	\sum_{|\alpha |\leq 2}  a_{\alpha,i}(t_k,\Zc(t,t_k)) \d_z^\alpha , &
\Zc(t,t_k)
  &:= z + \mv(t,t_k) + \Cv(t,t_k) \nabla_z . 	\label{eq:Gc.def}
\end{align}
\end{proposition}

\begin{proof}
The proof, which is given in \cite[Theorem 2.6]{lorig-pagliarani-pascucci-2}, relies on the fact that, for any $0 \leq t \leq t_k < \infty$ the operator $\Gc_i$ in \eqref{eq:Gc.def} satisfies
\begin{align}
\Pc_0(t,t_k) \Ac_{i}(t_k)
    &=  \Gc_i(t,t_k) \Pc_0(t,t_k) . \label{eq:PA=GP}
\end{align}
Using \eqref{eq:PA=GP}, as well as the semigroup property ${\Pc_0}(t_1,t_2) {\Pc_0}(t_2,t_3) = {\Pc_0}(t_1,{t_3})$, we have that
\begin{align}
&\Pc_0(t,t_1) \Ac_{i_1}(t_1) \Pc_0(t_1,t_2) \Ac_{i_2}(t_2) \cdots \Pc_0(t_{k-1},t_k) \Ac_{i_k}(t_k) \Pc_0(t_k,T) \phi(z) \\
	&=	\Gc_{i_1}(t,t_1) \Gc_{i_2}(t,t_2) \cdots \Gc_{i_k}(t,t_k)
       \Pc_0(t,t_1) \Pc_0(t_1,t_2) \cdots \Pc_0(t_{k-1},t_k) \Pc_0(t_k,T) \phi \\
	&=	\Gc_{i_1}(t,t_1) \Gc_{i_2}(t,t_2) \cdots \Gc_{i_k}(t,t_k) \Pc_0(t,T) \phi \\
	&=	\Gc_{i_1}(t,t_1) \Gc_{i_2}(t,t_2) \cdots \Gc_{i_k}(t,t_k) v_0(t,\,\cdot\,) , \label{eq:result}
\end{align}
where, in the last equality we have used $\Pc_0(t,T) \phi = v_0(t,\,\cdot\,)$.  Inserting \eqref{eq:result} into \eqref{eq:vn-explicit} yields \eqref{eq:un}.
\end{proof}

\noindent
Having obtained expressions for the functions $(v_n)_{n \geq 0}$ as differential operators $(\Lc_n)_{n \geq 0}$ acting on $v_0$, we define $\vb$, the \textit{$n$th order approximation of $v$}, as follows
\begin{align}
\vb_n
	:= \sum_{k=0}^n {v_k} . \label{eq:def-vbar}
\end{align}
Note that $\vb_n$ depends on the choice of $\zb$.  In general, if one is interested in the value of $v(t,z)$ a good choice for $\zb$ is $\zb(t) = z$.

%-----------------------------------------------------------------------------------
%
%       SECTION: 		volatility asymptotics
%
%-----------------------------------------------------------------------------------

\section{Implied volatility asymptotics}
\label{sec:imp-vol}
In this section, we show how to translate the price approximation developed in Section \ref{sec:price-asymptotics} into an approximation of Black implied volatilities associated with caplets.  The derivation below closely follows the derivation in \cite{lorig-pagliarani-pascucci-2}, where the authors develop an approximation for Black-Scholes implied volatilities associated with call options on equity.
\\[0.5em]
Throughout this section, we fix a QTS model \eqref{eq:dY}-\eqref{eq:r-def}, an initial date $t$, a reset date $T > t$, a settlement date $\Tb > T$, the initial values $(X_t = \log L_t^{T,\Tb}, Y_t) = (x, y)$ and a caplet payoff $\phi(X_T) = \tau (\ee^{X_T} - \ee^k)^+$. Our goal is to find an approximation of implied volatility for \textit{this particular caplet}.  To ease notation, we will sometimes hide the dependence on $(t, x, y; T, \Tb, k)$.  However, the reader should keep in mind that the implied volatility of the caplet under consideration does depend on $(t, x, y; T, \Tb, k)$, even if this is not explicitly indicated.  Below, we remind the reader of the \textit{Black model} and provide definitions of the \textit{Black price} and \textit{Black implied volatility}, which will be used throughout this section.
\\[0.5em]
In the \textit{Black model}, the dynamics of the simple forward rate $L^{T,\Tb} = \ee^X$ are given by
\begin{align}
\dd L^{T,\Tb}_t
	&=	\sig L^{T,\Tb}_t \dd \Wt_t , &
	&\text{and thus}&
\dd X_t
	&=	-\tfrac{1}{2}\sig^2 \dd t + \sig \dd \Wt_t , \label{eq:black-model}
\end{align}
where $\sig > 0$ is the \textit{Black volatility} and $\Wt$ is a scalar $(\Pbt,\Fb)$-Brownian motion. Equation \eqref{eq:black-model}  leads to the following definitions.

\begin{definition}
The $\Tb$-forward \textit{Black price} of a caplet, denoted $v^\BS$, is defined as follows
\begin{align}
v^\BS(t,x;T,\Tb,k,\sig)
	&:=	\tau \, \Ebt[  ( \ee^{X_T} - \ee^k )^+ | X_t = x ]
	=		\tau \, \Big( \ee^x \Phi(d_+) - \ee^{k } \Phi(d_-) \Big) , \label{eq:black-price}
\end{align}
where the dynamics of $X$ are given by \eqref{eq:black-model} and
\begin{align}
d_\pm
   &:= \frac{1}{\sig \sqrt{T-t}} \left(x-k  \pm \frac{\sig^2 (T-t)}{2}  \right) , &
\Phi(d)
  &:=  \int_{-\infty}^d \dd x \, \frac{1}{\sqrt{2\pi}} \ee^{-x^2/2}.
\end{align}
\end{definition}

\begin{definition}
The \textit{Black implied volatility} corresponding to the $\Tb$-forward price $v$ of a caplet is the unique positive solution $\sig$ of the equation
\begin{align}
v^{\BS}(t,x;T,\Tb,k,\sig) 
	&= v . \label{eq:iv-def}
\end{align}
where the Black price $v^\BS$ is given by \eqref{eq:black-price}.
\end{definition}

\noindent
Now, suppose that $v \equiv v(t,x,y;T,\Tb,k)$ is the $\Tb$-forward price of a caplet corresponding to a QTS model, where we have now indicated the dependence on the $\log$ strike $k$ explicitly.
As in Section \ref{sec:price-asymptotics}, we will seek an approximation of the implied volatility $\sig^\eps$ corresponding to $v^\eps$ by expanding $\sig^\eps$ in power of $\eps$.  Our approximation {of} $\sig$ will then be obtained by setting $\eps = 1$.  We have
\begin{align}
\sig^\eps
  &=  \sig_0 + \del \sig^\eps , &
\del \sig^\eps
  &=  \sum_{n=1}^\infty \eps^n \sig_n , \label{eq:I-expand}
\end{align}
where $(\sig_n)_{n \geq 0}$ are, at the moment, unknown.
Expanding the Black price $v^\BS(\sig^\eps)$ in powers of $\eps$ we obtain
\begin{align}
v^\BS(\sig^\eps)
  &=  v^\BS(\sig_0 + \del \sig^\eps) \\
  &=    \sum_{k=0}^\infty \frac{1}{k!}(\del \sig^\eps \d_\sig )^k v^\BS(\sig_0) \\
  %&=   v^\BS(\sig_0) +
            %\sum_{k=1}^\infty \frac{1}{k!} \( \sum_{n=1}^\infty \eps^n \sig_n \)^k \d_\sig^k v^\BS(\sig_0) \\
  &=    v^\BS(\sig_0) + 
            \sum_{k=1}^\infty \frac{1}{k!}  
            \sum_{n=1}^\infty \eps^n \sum_{I_{n,k}} \Big( \prod_{j=1}^k \sig_{i_j} \Big) \d_\sig^k v^\BS(\sig_0) \\
  &=    v^\BS(\sig_0) + 
            \sum_{n=1}^\infty \eps^n  \sum_{k=1}^\infty \frac{1}{k!}  
            \sum_{ I_{n,k}} \Big( \prod_{j=1}^k \sig_{i_j} \Big) \d_\sig^k v^\BS(\sig_0) \\
  &=    v^\BS(\sig_0) + 
            \sum_{n=1}^\infty \eps^n \bigg( \sig_n \d_\sig + \sum_{k=2}^\infty \frac{1}{k!}  
            \sum_{ I_{n,k}} \Big( \prod_{j=1}^k \sig_{i_j} \Big) \d_\sig^k  \bigg) v^\BS(\sig_0) ,
\end{align}
where $I_{n,k}$ is given by \eqref{eq:Ink}.  Inserting the expansions for $v^\eps$ and $v^\BS(\sig^\eps)$ into the equation $v^\eps = v^\BS(\sig^\eps)$ and collecting terms of like order in $\eps$ we obtain
\begin{align}
&\Oc(\eps^0)&
v_0
  &=  v^\BS(\sig_0) , \label{eq:v0=expression} \\
&\Oc(\eps^n)&
v_n
  &=  \bigg( \sig_n \d_\sig + \sum_{k=2}^\infty \frac{1}{k!}  \sum_{ I_{n,k}} \Big( \prod_{j=1}^k \sig_{i_j} \Big) \d_\sig^k  \bigg) v^\BS(\sig_0) . \label{eq:vn=expression}
\end{align}
Now, from \eqref{eq:v0-explicit} and \eqref{eq:black-price} we have
\begin{align}
v_0
  &=  v^\BS \left( \sqrt{ \Cv_{1,1}(t,T)/(T-t) } \right) ,
\end{align}
where $\Cv$ is defined in \eqref{eq:m-and-C}.  Thus, it follows from \eqref{eq:v0=expression} that
\begin{align}
\sig_0
  &=  \sqrt{ \Cv_{1,1}(t,T)/(T-t) } .\label{eq:sig-0}
\end{align}
Having identified $\sig_0$, we can use \eqref{eq:vn=expression} to obtain $\sig_n$ recursively for every $n \geq 1$.  We have
\begin{align}
\sig_n
  &=  \frac{1}{\d_\sig v^\BS(\sig_0)} \bigg( v_n - \sum_{k=2}^\infty \frac{1}{k!}  \sum_{ I_{n,k}} \Big( \prod_{j=1}^k \sig_{i_j} \Big) \d_\sig^k v^\BS(\sig_0) \bigg) . \label{eq:sig-n}
\end{align}
Using the expression given in \eqref{eq:un} for $v_n$, one can show that $\sig_n$ is an $n$th order polynomial in $\log$-moneyness $k-x$ with coefficients that depend on $(t,T)$; see \cite[Section 3]{lorig-pagliarani-pascucci-2} for details.  We provide explicit expressions for $\sig_0$, $\sig_1$, and $\sig_2$ for $d = 1$ in Appendix \ref{sec:explicit-expressions}.
\\[0.5em]
We now define our $n$\textit{th order approximation of implied volatility} as
\begin{align}
\sigb_n
	&:=	\sum_{j=0}^n \sig_j , \label{eq:sigma-bar}
\end{align}
where $\sig_j$ is given by \eqref{eq:sig-n}.  
If we set $\zb(t) = (x,y)$ in the price approximation, then the corresponding implied volatility approximation \eqref{eq:sigma-bar} satisfies the following asymptotic accuracy result
\begin{align}
|\sig(t,x,y;T,\Tb,k)-\sigb_n(t,x,y;T,\Tb,k)| 
	&= \Oc((T-t)^{(n+1)/2}), &
		&\text{as}&
(T,k) \to (t,x), \label{eq:accuracy}
\end{align}
within the parabolic region $\{(T,k) : |k-x| \leq \ell \sqrt{T-t} \}$ for some $\ell > 0$. The proof of \eqref{eq:accuracy} is a direct consequence of \cite[Theorem 3.10]{VIX}.

%-----------------------------------------------------------------------------------
%
%       SECTION: 		Examples
%
%-----------------------------------------------------------------------------------

\section{Numerical example: Quadratic Ornstein-Uhlenbeck model}
\label{sec:examples} 
Throughout this section, we consider a QTS model, whose dynamics are as follows
\begin{align}
\dd Y_t
  &=  \kappa ( \theta - Y_t) \dd t + \del \dd W_t , &
R_t = r(Y_t)
  &=  q + Y_t^2, \label{eq:quadratic-vasicek-def}
\end{align}
where the constants {$\kappa,\del$ are positive and $\theta,q$ are nonnegative}. 
Noting that $Y$ is an Ornstein-Uhlenbeck process,  
we refer to the model \eqref{eq:quadratic-vasicek-def} as the \textit{Quadratic Ornstein-Uhlenbeck} (QOU) model. 

\begin{remark}
If we consider the special case $\theta = q = 0$.
then we have by It\^o's lemma that
\begin{align}
\dd R_t 
  & = 2\kappa \Big( \frac{\del^2}{2\kappa}- R_t \Big) \dd t + 2\del \sqrt{R_t} \dd W_t . \label{eq:cir}
\end{align}
Note that \eqref{eq:cir} is a \textit{Cox-Ingersoll-Ross} (CIR) process with a mean $\frac{\del^2}{2\kappa}$, rate of mean-reversion $2 \kappa$ and volatility $2 \del$.  Thus, the QOU model contains as a special case, some (but not all) CIR short-rate models.
\end{remark}

\noindent
Comparing \eqref{eq:quadratic-vasicek-def} with \eqref{eq:dY} and \eqref{eq:r-def} we obtain  
\begin{align}
     \lam & = \kappa \theta, & \Lam & = -\kappa, & \Sig & = \delta, & \Xi &= 1.
\end{align}
Next, we can obtain from \eqref{eq:Fode}, \eqref{eq:Gode}, and \eqref{eq:Hode} that $(F,G,H)$ satisfies the following system of ODEs
\begin{align}
\left.
\begin{aligned}
\d_t F(t;T,\nu,\Omega) & = \tfrac{1}{2}\del^2 G^{2}(t;T,\nu,\Omega)-\del^2H(t;T,\nu,\Omega)-\kappa \theta G(t;T,\nu,\Omega)-q,
&
F(T;T,\nu,\Omega) & = 0 ,
\\
\d_t G(t;T,\nu,\Omega) & = \big(2\del^2 H(t;T,\nu,\Omega)+ \kappa\big) G(t;T,\nu,\Omega)-2\kappa\theta H(t;T,\nu,\Omega) ,
&
G(T;T,\nu,\Omega) & = -\nu ,
\\
\d_t H(t;T,\nu,\Omega) & =  2\del^2 H^2(t;T,\nu,\Omega)+2\kappa H(t;T,\nu,\Omega) - 1 ,
&
H(T;T,\nu,\Omega) & = -\Omega .
\end{aligned}
\right \}
\label{eq:qts-vasicek-ode}
\end{align}
Solving \eqref{eq:qts-vasicek-ode}, we obtain
\begin{align}
    F(t;T,\nu,\Omega) & = \int_{T}^{t} \dd s \, \Big( \tfrac{1}{2}\del^2 G^2(s;T,\nu,\Omega) -\del^2 H(s;T,\nu,\Omega) -\kappa\theta G(s;T,\nu,\Omega) -q \Big), \label{eq:F-cir-qts}
\end{align}
\begin{align}
    G(t;T,\nu,\Omega) & = -\frac{Q_1(T-t) \nu + Q_2(T-t) \Omega + Q_3(T-t)}{Q_4(T-t)\Omega+Q_5(T-t)},
    &
    H(t;T,\nu,\Omega) & = -\frac{Q_6(T-t)\Omega+Q_7(T-t)}{Q_4(T-t)\Omega+Q_5(T-t)}, \label{eq:GH-cir-qts}
\end{align}
where the functions $Q_i(t)$ for $i\in \{1,2,\ldots,7\}$ are given by 
\begin{align}
    Q_1(t) & := 2 \mu  \ee^{\tfrac{1}{2}\mu t},
    &
    Q_2(t) &:= \frac{8\del^2}{\mu } \left(\ee^{\tfrac{1}{2}\mu t}-1\right)^2 \left(\frac{- \kappa^2\theta }{\del^2 }\right)-\frac{\kappa\theta  Q_4(t)}{\del^2 },
    \\
    Q_3(t) & :=  -\frac{\kappa\theta }{\del^2}\left(\frac{\kappa}{\mu } Q_7\left(\frac{t}{2}\right)Q_5\left(\frac{t}{2}\right)-Q_1(t)+Q_5(t)\right),
    &
    Q_4(t)& := 4\del^2  (1-\ee^{\mu t}),
    \\
    Q_5(t) & := \mu  (\ee^{\mu t}+1)+2\kappa (\ee^{\mu t}-1),
    &
    Q_6(t) & := \mu  (\ee^{\mu t}+1)-2\kappa (\ee^{\mu t}-1),
    \\
    Q_7(t) & := 2 (1-\ee^{\mu t}),
    &
    \mu &  := 2 \sqrt{\kappa^2 + 2\del^2 }.
\end{align}
\noindent
Next, from \eqref{eq:A-1d} we have the form of the generator
\begin{align}
\Act(t) 
  &=  c(t,x,y) (\d_x^2 - \d_x ) + f(t,x,y) \d_y + g(t,x,y) \d_y^2  + h(t,x,y) \d_x \d_y, 
\end{align}
where the functions $c$, $f$, $g$ and $h$ are given by 
\begin{align}
c(t,x,y)&  = \tfrac{1}{2} \del^2 \Big( 1 + \frac{\ee^{-x}}{\tau} \Big)^2 \Big( \Gf(t;\Tb) - \Gf(t;T) + 2 \Big( \Hf(t;\Tb) - \Hf(t;T) \Big) y \Big)^2,
  \\
f(t,x,y)
  & = \kappa\theta-\kappa y  - \del^2\Big( \Gf(t;\Tb) +2 \Hf(t;\Tb) y \Big), \\
g(t,x,y)
  &= \tfrac{1}{2}\del^2 , \\
h(t,x,y)
  &=  \del^2 \Big( 1 + \frac{\ee^{-x}}{\tau} \Big) \Big( \Gf(t;\Tb) - \Gf(t;T) + 2 \Big( \Hf(t;\Tb) - \Hf(t;T) \Big) y \Big).
\end{align}
Introducing the notation 
\begin{align} \chi_{i,j}(t,x,y) & := \frac{1}{i! j!}\d_x^i \d_{y}^j \chi(t,x,y)  &
	&\text{where}&  \chi & \in \{c,f,g,h\},
\end{align}
the explicit implied volatility approximation $\sigb_n$ can now be computed up to order $n=2$ using the formulas in Appendix \ref{sec:explicit-expressions}. We have
\begin{align}
    \sig_0 & = \sqrt{\frac{2}{T-t}\int_{t}^T \dd s \, c_{0,0}(s,x,y_2)},
\\
    \sig_1 & = \frac{(k-x)}{(T-t)^2\sig^3_0}\Big(2\int_{t}^T\dd s \, c_{1,0}(s,x,y_2)\int_{t}^s \dd q \, c_{0,0}(q,x,y_2)+ \int_{t}^T \dd s \, c_{0,1}(s,x,y_2)\int_{t}^s \dd q \, h_{0,0}(q,x,y_2)\Big)
    \\
     &\quad + \frac{1}{2(T-t)\sig_0}\int_{t}^T \dd s \, c_{0,1}(s,x,y_2)\Big(2\int_{t}^s \dd q \, f_{0,0}(q,x,y_2)+ \int_{t}^s \dd q \, h_{0,0}(q,x,y_2)\Big) ,
\end{align}
where we have omitted the 2nd order term $\sig_2$ due to its considerable length.
\\[0.5em]
{In Figures \ref{fig:qts-vasicek-iv} and \ref{fig:qts-cir-iv}}, using different parameters for $(\kappa,\theta,\del,q,y)$, we plot our explicit approximation of implied volatility $\sigb_n$ up to order $n=2$ as a function of $\log$-moneyness $k-x$ with $t=0$ and $\Tb = 2$ fixed and with reset date ranging over $T = \{\frac{1}{64},\frac{1}{32},\frac{1}{16},\frac{1}{8}\}$.
{For comparison, we also plot the ``exact'' implied volatility $\sig$, which   can be computed using $\Tb$-forward caplet prices using \eqref{eq:v-explicit} and inverting the Black formula \eqref{eq:black-price} numerically. }
In both figures, we observe that the second order approximation $\sigb_2$ accurately matches the level, slope, and convexity of the exact implied volatility $\sig$ near-the-money for all four reset dates.
\\[0.5em]
{In Figures \ref{fig:qts-vasicek-err} and \ref{fig:qts-cir-err}}, using the same values for $(\kappa,\theta,\del,q,y)$ as in Figures \ref{fig:qts-vasicek-iv} and \ref{fig:qts-cir-iv}, respectively, we plot the absolute value of the relative error of our second order approximation $|\sigb_2-\sig|/\sig$ as a function of $\log$-moneyness $k-x$ and reset date $T$.
Consistent with the asymptotic accuracy results \eqref{eq:accuracy}, we observe that the {errors decrease} as we approach the origin in both directions of $k-x$ and $T$.

%-----------------------------------------------------------------------------------
%
%         Appendix
%
%-----------------------------------------------------------------------------------

\appendix

\section{Explicit expressions for \texorpdfstring{$\sig_0$}{}, \texorpdfstring{$\sig_1$}{} and \texorpdfstring{$\sig_2$}{}} 
\label{sec:explicit-expressions}
In this appendix we give the expressions for the implied volatility approximation using \eqref{eq:sig-0} and \eqref{eq:sig-n} explicitly up to second order for $d=1$ in terms of the coefficients $c$,$f$,$g$, and $h$ of $\Act$, given in \eqref{eq:A-1d}, by performing Taylor's series expansion of the coefficients around $\zb(t) = (x,y)$.  To ease the notation, we define
\begin{align}
 \chi_{i,j}(t) \equiv \chi_{i,j}(t,x,y) &   = \frac{\d^i_x\d^j_y\chi(t,x,y)}{i!j!}, &  \chi & \in \{c,f,g,h\}. \label{eq:chi-ij}
\end{align}
The zeroth order term $\sig_0$ is given by
\begin{align}
        \sig_0 &  = \sqrt{\frac{2}{T-t}\int_{t}^T \dd s \, c_{0,0}(s) } .
                %& d=\{1,2\},
\end{align}
Next, let us define 
\begin{align}
\mathscr{H}_n(\Theta) &:= \Big(\frac{-1}{\sigma_0\sqrt{2(T-t)}}\Big)^n \mathsf{H}_n(\Theta), & 
\Theta &:= \frac{x-k-\frac{1}{2}\sigma^2_0(T-t)}{\sig_0\sqrt{2(T-t)}}.
\end{align}
where $\mathsf{H}_n(\Theta)$ is the  $n$th-order \emph{Hermite polynomial}.
Then the first order term $\sig_1$ is given by
\begin{align}
    %\sig_1 & = \sig_{1,0}, & \sig_{2} & = \sig_{2,0}, & d = 1,
    %\\
\sig_1 & = \sig_{1,0}+\sig_{0,1},  
                %& \sig_{2} & = \sig_{2,0}+\sig_{1,1}+\sig_{0,2}, & d = 1,
\end{align}
where $ \sig_{1,0}$ and $\sig_{0,1}$ are given by
\begin{align}
    \sig_{1,0} & = \frac{1}{(T-t)\sig_0}\int_{t}^T\dd s \, c_{1,0}(s)\int_{t}^s \dd q \, c_{0,0}(q)  \Big(2\mathscr{H}_1(\Theta)-1\Big),  \\
    \sig_{0,1} & = \frac{1}{(T-t)\sig_0}\int_{t}^T \dd s \, c_{0,1}(s)\Big(\int_{t}^s \dd q \, f_{0,0}(q)+ \int_{t}^s \dd q \, h_{0,0}(q) \mathscr{H}_1(\Theta)\Big) .
\end{align}
 Lastly, the second order term $\sig_2$ is given by
\begin{align}
\sig_{2}
    & = \sig_{2,0}+\sig_{1,1}+\sig_{0,2}, 
\end{align}
where the terms $\sig_{2,0}$, $\sig_{1,1}$, $\sig_{0,2}$ are given by
\begin{align}
    \sig_{2,0} & = \frac{1}{(T-t)\sig_0}\bigg(\frac{1}{2}\int_{t}^T \dd s \, c_{2,0}(s) \Big((\int_{t}^s \dd q \, c_{0,0}(q) )^2 (4\mathscr{H}_2(\Theta)-4\mathscr{H}_1(\Theta)+1) + 2 \int_{t}^s \dd q \, c_{0,0}(q)  \Big)
    \\
    & \quad + \int_{t}^T \dd s_1 \, \int_{s_1}^T \dd s_2 \, c_{1,0}(s_1)c_{1,0}(s_2)\Big(\int_{t}^{s_1} \dd q_1 \, c_{0,0}(q_1) \int_{t}^{s_2} \dd q_2 \, c_{0,0}(q_2)
    \\
    & \quad \times 
    \big(4\mathscr{H}_4(\Theta)-8\mathscr{H}_3(\Theta)+5\mathscr{H}_2(\Theta)-\mathscr{H}_1(\Theta)\big) + \int_{t}^{s_1} \dd q_1 \, c_{0,0}(q_1)  \Big(6\mathscr{H}_2(\Theta)-6\mathscr{H}_1(\Theta)+1)\Big)\bigg) 
    \\
    & \quad -\frac{\sig^2_{1,0}}{2}\Big((T-t)\sig_0 (\mathscr{H}_2(\Theta)-\mathscr{H}_1(\Theta))
    + \frac{1}{\sig_0} \Big), \\
    \sig_{1,1} & = \frac{1}{(T-t)\sig_0}\Bigg(\frac{1}{2}\int_{t}^T \dd s \, c_{1,1}(s) \bigg(2\int_{t}^{s} \dd q_1 \, c_{0,0}(q_1) \int_{t}^{s} \dd q_2 \, h_{0,0}(q_2) \mathscr{H}_2(\Theta) 
    \\
    & \quad + \int_{t}^{s} \dd q_1 \, c_{0,0}(q_1)(2\int_{t}^{s} \dd q_2 \, f_{0,0}(q_2)-\int_{t}^{s} \dd q_2 \, h_{0,0}(q_2))\mathscr{H}_1(\Theta) 
    \\
    & \quad -\int_{t}^{s} \dd q_1 \, c_{0,0}(q_1)\int_{t}^{s} \dd q_2 \, f_{0,0}(q_2)+\int_{t}^{s} \dd q_1 \, h_{0,0}(q_1)\bigg) 
    \\
    & \quad + \int_{t}^T \dd s_1 \,\int_{s_1}^{T} \dd s_2 \, c_{1,0}(s_1)c_{0,1}(s_2) \bigg( 2\int_{t}^{s_1} \dd q_1 \, c_{0,0}(q_1)\int_{t}^{s_2} \dd q_2 \, h_{0,0}(q_2)\mathscr{H}_4(\Theta) 
        \\
        & \quad +\int_{t}^{s_1} \dd q_1 \, c_{0,0}(q_1) \Big(2\int_{t}^{s_2} \dd q_2 \, f_{0,0}(q_2) - 3\int_{t}^{s_2} \dd q_2 \, h_{0,0}(q_2)\Big)\mathscr{H}_3(\Theta) \\
        & \quad +(\int_{t}^{s_1} \dd q \, c_{0,0}(q)(\int_{t}^{s_2} \dd q \, h_{0,0}(q)-3 \int_{t}^{s_2} \dd q \, f_{0,0}(q))+\int_{t}^{s_1} \dd q \, h_{0,0}(q)\Big) \mathscr{H}_2(\Theta) 
        \\
        & \quad + \Big(\int_{t}^{s_1} \dd q_1 \, c_{0,0}(q_1)\int_{t}^{s_2} \dd q_2 \, f_{0,0}(q_2)-\int_{t}^{s_1} \dd q_1 \, h_{0,0}(q_1)\Big)\mathscr{H}_1(\Theta) \bigg)
        \\
        & \quad + \int_{t}^T \dd s_1 \,\int_{s_1}^{T} \dd s_2 \, c_{0,1}(s_1)c_{1,0}(s_2) \bigg(2 \int_{t}^{s_1} \dd q_1 \, h_{0,0}(q_1) \int_{t}^{s_2} \dd q_2 \, c_{0,0}(q_2) \mathscr{H}_4(\Theta) 
        \\
        & \quad +  \Big(2\int_{t}^{s_1} \dd q_1 \, f_{0,0}(q_1)-3\int_{t}^{s_1} \dd q_1 \, h_{0,0}(q_1)\Big)\int_{t}^{s_2} \dd q_2 \, c_{0,0}(q_2)\mathscr{H}_3(\Theta) 
        \\
        & \quad + \Big(\big(\int_{t}^{s_1} \dd q_1 \, h_{0,0}(q_1)-3 \int_{t}^{s_1} \dd q_1 \, f_{0,0}(q_1)\big)\int_{t}^{s_2} \dd q_2 \, c_{0,0}(q_2) + 3 \int_{t}^{s_1} \dd q_1 \, h_{0,0}(q_1)\Big)\mathscr{H}_2(\Theta) 
        \\
        & \quad + \Big(\int_{t}^{s_1} \dd q_1 \, f_{0,0}(q_1)(2 +\int_{t}^{s_2} \dd q_2 \, c_{0,0}(q_2))-2\int_{t}^{s_1} \dd q_1 \, h_{0,0}(q_1)\Big) \mathscr{H}_1(\Theta) - \int_{t}^{s_1} \dd q_1 \, f_{0,0}(q_1) \bigg)
        \\
        & \quad + \int_{t}^T \dd s_1 \,\int_{s_1}^{T} \dd s_2 \,f_{1,0}(s_1)c_{0,1}(s_2)\int_{t}^{s_1} \dd q_1 \, c_{0,0}(q_1)\Big(2\mathscr{H}_1(\Theta) -1\Big)
        \\
        & \quad + 2\int_{t}^T \dd s_1 \,\int_{s_1}^{T} \dd s_2 \,h_{1,0}(s_1)c_{0,1}(s_2)\int_{t}^{s_1} \dd q_1 \, c_{0,0}(q_1)\Big(2\mathscr{H}_2(\Theta)  -\mathscr{H}_1(\Theta)\Big) \Bigg)
        \\
        & \quad -\sig_{1,0}\sig_{0,1}\Big((T-t)\sig_0 (\mathscr{H}_2(\Theta)-\mathscr{H}_1(\Theta))
    + \frac{1}{\sig_0} \Big), \\
    \sig_{0,2} & = \frac{1}{(T-t)\sig_0}\Bigg(\frac{1}{2}\int_{t}^T \dd s \, c_{0,2}(s) \bigg(\Big(\int_{t}^{s} \dd q \, h_{0,0}(q)\big)^2 \mathscr{H}_2(\Theta) + 2\int_{t}^{s} \dd q_1 \, h_{0,0}(q_1)\int_{t}^{s} \dd q_2 \, f_{0,0}(q_2)\mathscr{H}_1(\Theta) 
    \\
    & \quad + (\int_{t}^{s} \dd q \, f_{0,0}(q))^2 + 2\int_{t}^{s} \dd q \, g_{0,0}(q)\bigg)
    \\
    & \quad + \int_{t}^T \dd s_1 \,\int_{s_1}^{T} \dd s_2 \, c_{0,1}(s_1)c_{0,1}(s_2) \bigg( \int_{t}^{s_1} \dd q_1 \, h_{0,0}(q_1) \int_{t}^{s_2} \dd q_2 \, h_{0,0}(q_2) \mathscr{H}_4(\Theta)
    \\
    & \quad +  \Big(\int_{t}^{s_1} \dd q_1 \, f_{0,0}(q_1) \int_{t}^{s_2} \dd q_2 \, h_{0,0}(q_2) + \int_{t}^{s_1} \dd q_1 \, h_{0,0}(q_1)\int_{t}^{s_2} \dd q_2 \, f_{0,0}(q_2) \\
    & \quad -\int_{t}^{s_1} \dd q_1 \, h_{0,0}(q_1) \int_{t}^{s_2} \dd q_2 \, h_{0,0}(q_2) \Big)  \mathscr{H}_3(\Theta)
    \\
    & \quad + \Big(2\int_{t}^{s_1} \dd q \, g_{0,0}(q) + \int_{t}^{s_1} \dd q_1 \, f_{0,0}(q_1) \int_{t}^{s_2} \dd q_2 \, f_{0,0}(q_2) -\int_{t}^{s_1} \dd q_1 \, f_{0,0}(q_1) \int_{t}^{s_2} \dd q_2 \, h_{0,0}(q_2)
    \\
    & \quad - \int_{t}^{s_2} \dd q_1 \, f_{0,0}(q_1) \int_{t}^{s_1} \dd q_2 \, h_{0,0}(q_2)\Big)\mathscr{H}_2(\Theta) 
    \\
    & \quad -\Big(2\int_{t}^{s_1} \dd q \, g_{0,0}(q)+\int_{t}^{s_1} \dd q_1 \, f_{0,0}(q_1) \int_{t}^{s_2} \dd q_2 \, f_{0,0}(q_2)\Big) \mathscr{H}_1(\Theta)\bigg)
    \\
    & \quad + \int_{t}^T \dd s_1 \,\int_{s_1}^{T} \dd s_2 \, f_{0,1}(s_1)c_{0,1}(s_2)\Big(\int_{t}^{s_1} \dd q \, h_{0,0}(q)\mathscr{H}_1(\Theta)+\int_{t}^{s_1} \dd q \, f_{0,0}(q)\Big)
    \\
    & \quad + \int_{t}^T \dd s_1 \,\int_{s_1}^{T} \dd s_2 \, h_{0,1}(s_1)c_{0,1}(s_2)\Big(\int_{t}^{s_1} \dd q \, h_{0,0}(q) \mathscr{H}_2(\Theta) + \int_{t}^{s_1} \dd q \, f_{0,0}(q)\mathscr{H}_1(\Theta)\Big)\Bigg)
    \\
    & \quad -\frac{\sig^2_{0,1}}{2}\Big((T-t)\sig_0 (\mathscr{H}_2(\Theta)-\mathscr{H}_1(\Theta))
    + \frac{1}{\sig_0} \Big).
\end{align}
Note that, although $\mathscr{H}_3(\Theta)$ and $\mathscr{H}_4(\Theta)$ appear in the expressions for $\sig_{2,0}, \sig_{1,1}$ and $\sig_{0,2}$, the 3rd and 4th order terms in $k-x$ cancel the 3rd and 4th order terms resulting from $\{\sig_{1,0}^2\mathscr{H}_2(\Theta), \sig_{1,0}^2\mathscr{H}_1(\Theta) \}$,
$ \{\sig_{0,1}\sig_{1,0}\mathscr{H}_2(\Theta), \sig_{0,1}\sig_{1,0}\mathscr{H}_1(\Theta)\}$, and $\{ \sig_{0,1}^2\mathscr{H}_2(\Theta),\sig_{0,1}^2\mathscr{H}_1(\Theta)\}$, respectively, resulting in a second order implied volatility expansion that is quadratic in $k-x$.

\bibliography{bibliography}

\clearpage

\begin{figure}
\centering
\begin{tabular}{cc}
\includegraphics[width=0.45\textwidth]{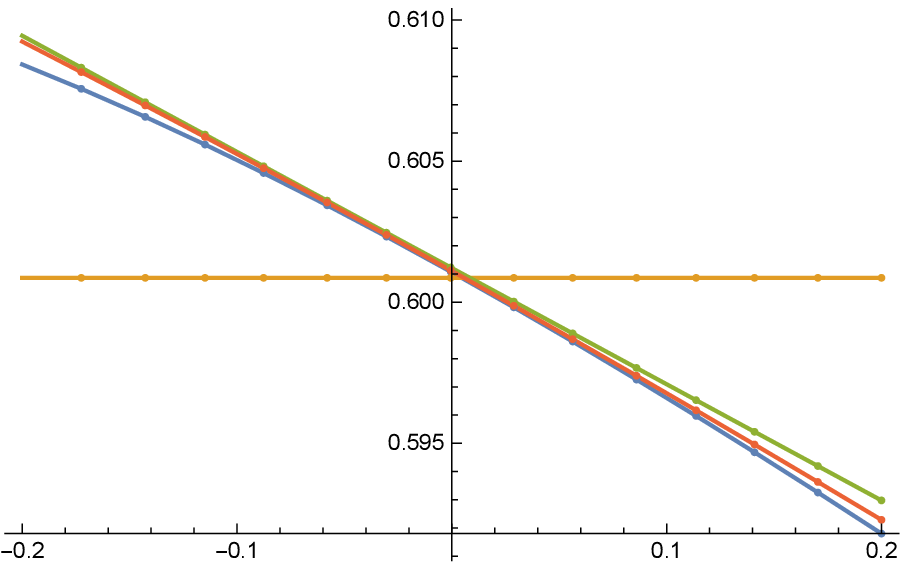}&
\includegraphics[width=0.45\textwidth]{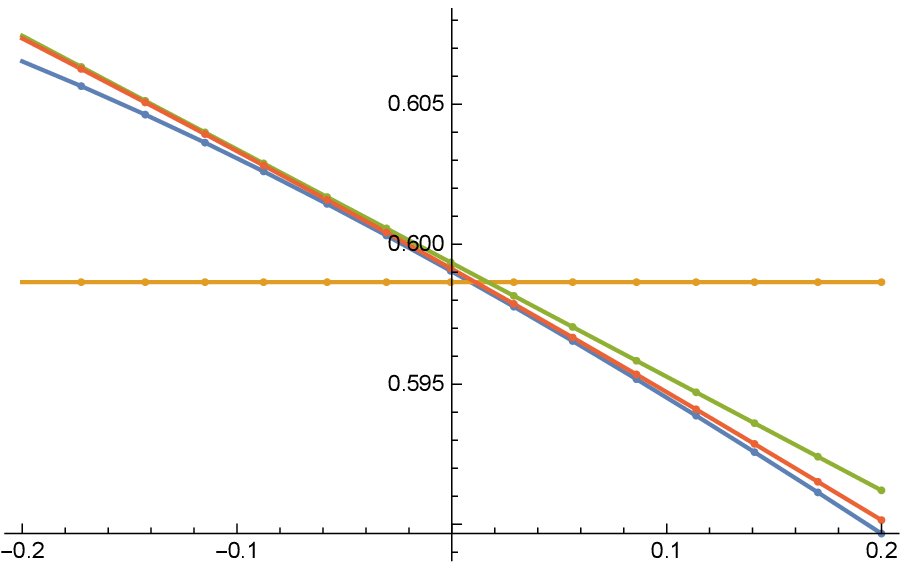}\\
$T = \frac{1}{64}$ & $T = \frac{1}{32}$ \\[1em]
\includegraphics[width=0.45\textwidth]{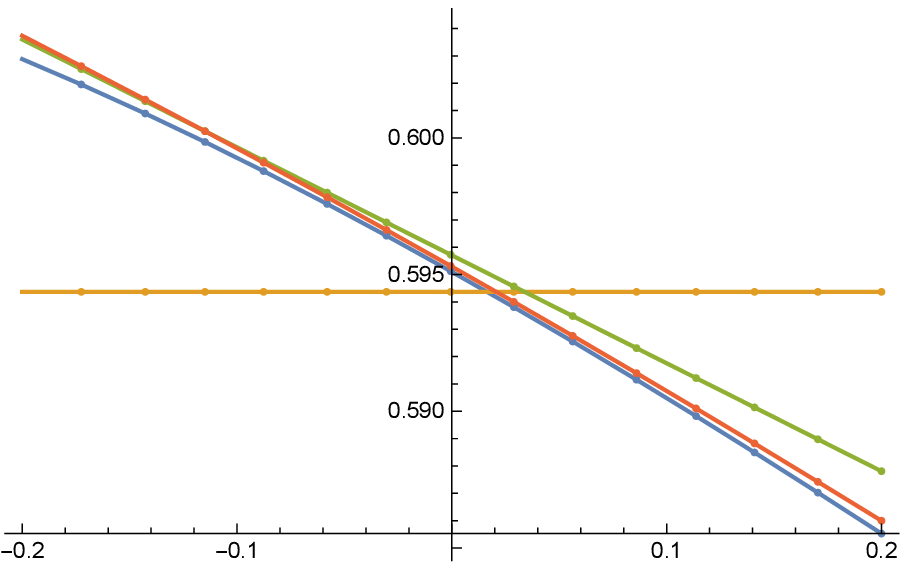}&
\includegraphics[width=0.45\textwidth]{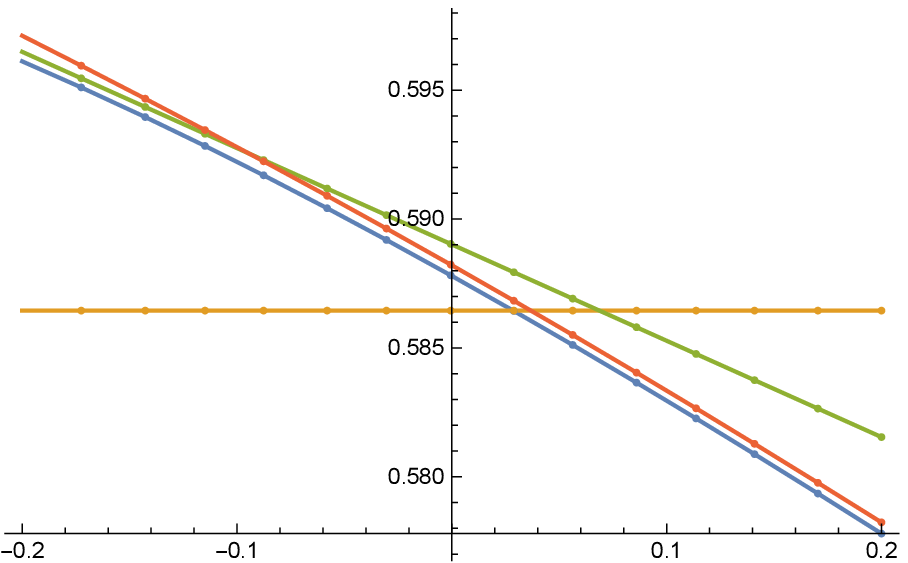}\\
$T = \frac{1}{16}$ & $T = \frac{1}{8}$
\end{tabular}
\caption{For the QOU model described in Section \ref{sec:examples}, we plot exact implied volatility $\sig$ and approximate implied volatility $\sigb_n$ up to order $n=2$ as a function of $\log$-moneyness $k-x$ with the initial date and settlement date of the caplet is fixed at $t = 0$ and $\Tb = 2$, respectively, and with the reset date of the caplet taking the following values $T = \{\frac{1}{64}, \frac{1}{32}, \frac{1}{16}, \frac{1}{8}\}$.  The zeroth, first, and second order approximate implied volatilities correspond to the orange, green and red curves, respectively, and the blue curve corresponds to the exact implied volatility. 
%In this setting, we assume that $q = 0$, and 
The following parameters remained fixed: $q=0$, $\kappa = 0.9$, {$\theta = \frac{0.25}{0.9}$, {$\delta = 0.2$}}, $y = \sqrt{0.08}$.}
\label{fig:qts-vasicek-iv}
\end{figure}

\begin{figure}
\centering
\begin{tabular}{cc}
\includegraphics[width=0.45\textwidth]{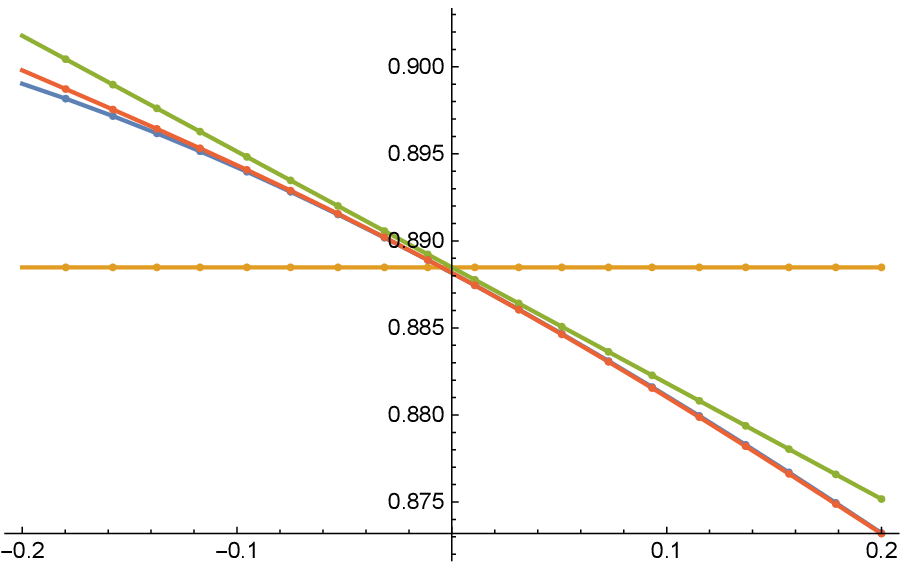}&
\includegraphics[width=0.45\textwidth]{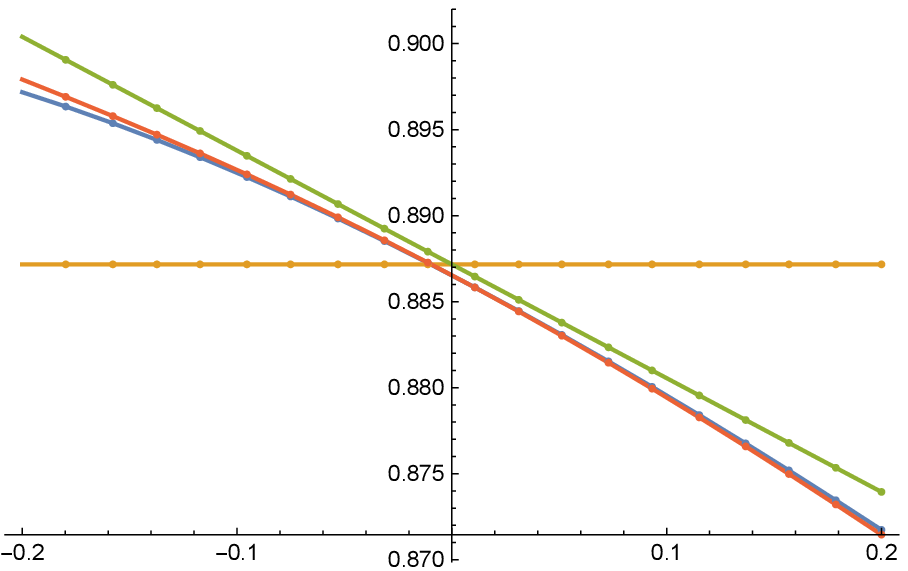}\\
$T = \frac{1}{64}$ & $T = \frac{1}{32}$ \\[1em]
\includegraphics[width=0.45\textwidth]{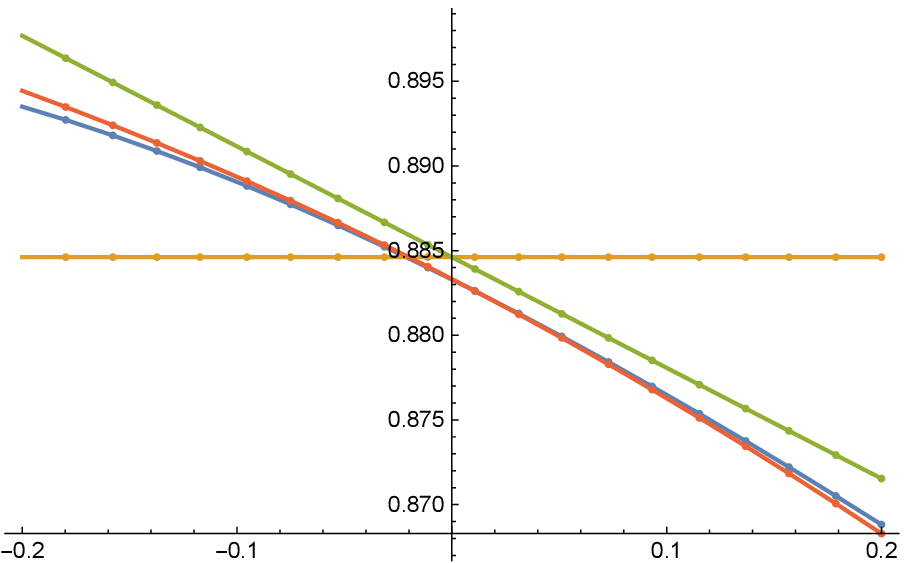}&
\includegraphics[width=0.45\textwidth]{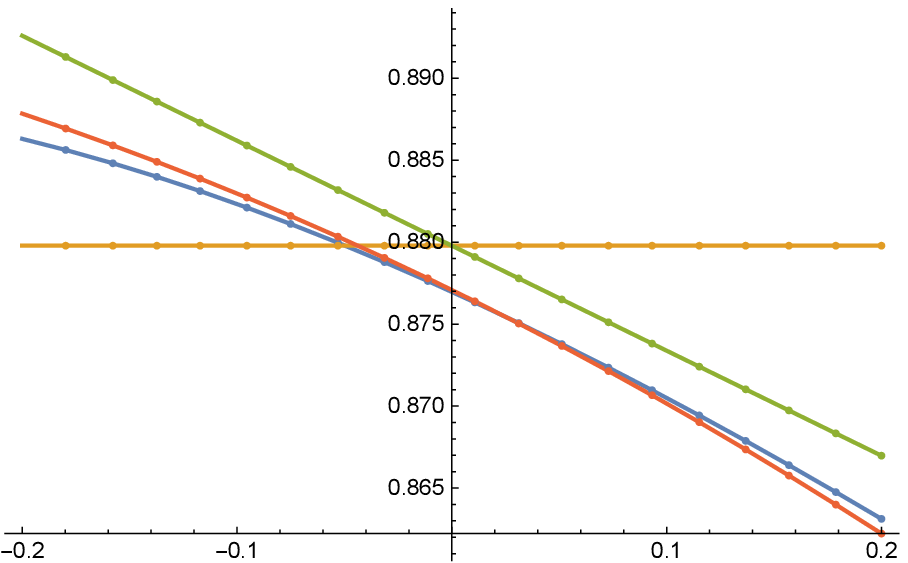}\\
$T = \frac{1}{16}$ & $T = \frac{1}{8}$
\end{tabular}
\caption{For the QOU model described in Section \ref{sec:examples}, we plot the exact implied volatility $\sig$ and approximate implied volatility $\sigb_n$ up to order $n=2$ as a function of $\log$-moneyness $k-x$ with the initial date and settlement date of the caplet is fixed at $t = 0$ and $\Tb = 2$, respectively, and with the reset date of the caplet taking the following values $T = \{\frac{1}{64}, \frac{1}{32}, \frac{1}{16}, \frac{1}{8}\}$.  The zeroth, first, and second order approximate implied volatilities correspond to the orange, green and red curves, respectively, and the blue curve corresponds to the exact implied volatility.  The following parameters remained fixed: 
$\kappa = 0.045$, {$\delta = \sqrt{0.035}$}, $y = \sqrt{0.08}$, $\theta = q = 0$.}
 %to ensure that the QTS model \eqref{eq:quadratic-vasicek-def} is indeed a CIR model under ATS settings.
\label{fig:qts-cir-iv}
\end{figure}

\begin{figure}
\centering
\begin{tabular}{c}
\includegraphics[width=0.8\textwidth]{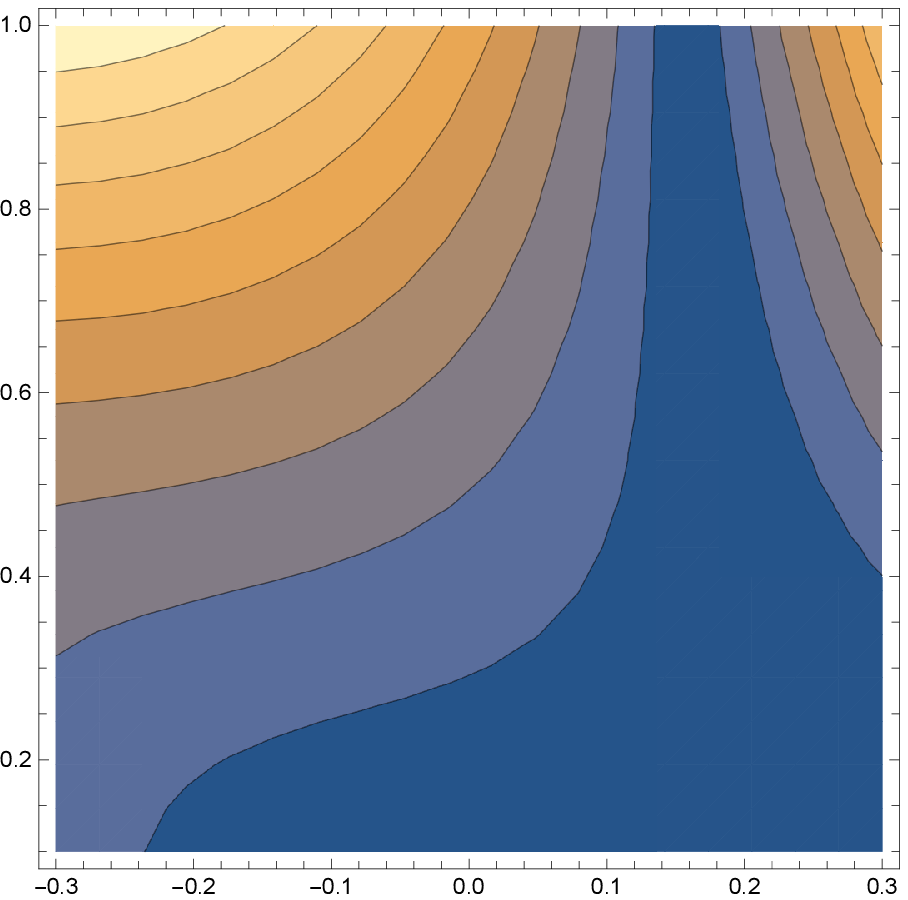}
\end{tabular}
\caption{
For the QOU model described in Section \ref{sec:examples}, we plot the absolute value of the relative error of our second order implied volatility approximation $|\sigb_2 - \sig|/\sig$ as a function of
$\log$-moneyness $k-x$ and caplet reset date $T$. The horizontal axis represents $\log$-moneyness $k -x$ and the vertical axis represents caplet reset date $T$. Ranging from
darkest to lightest, the regions above represent relative errors in increments of {$0.002$  from $< 0.002$ to $> 0.018$}. The initial date and settlement date of the caplet is fixed at $t = 0$ and $\Tb = 2$, respectively. 
The following parameters remained fixed:  $q=0$, $\kappa = 0.9$, {$\theta = \frac{0.25}{0.9}$, {$\delta = 0.2$}}, $y = \sqrt{0.08}$.}
\label{fig:qts-vasicek-err}
\end{figure}

\begin{figure}
\centering
\begin{tabular}{c}
\includegraphics[width=0.8\textwidth]{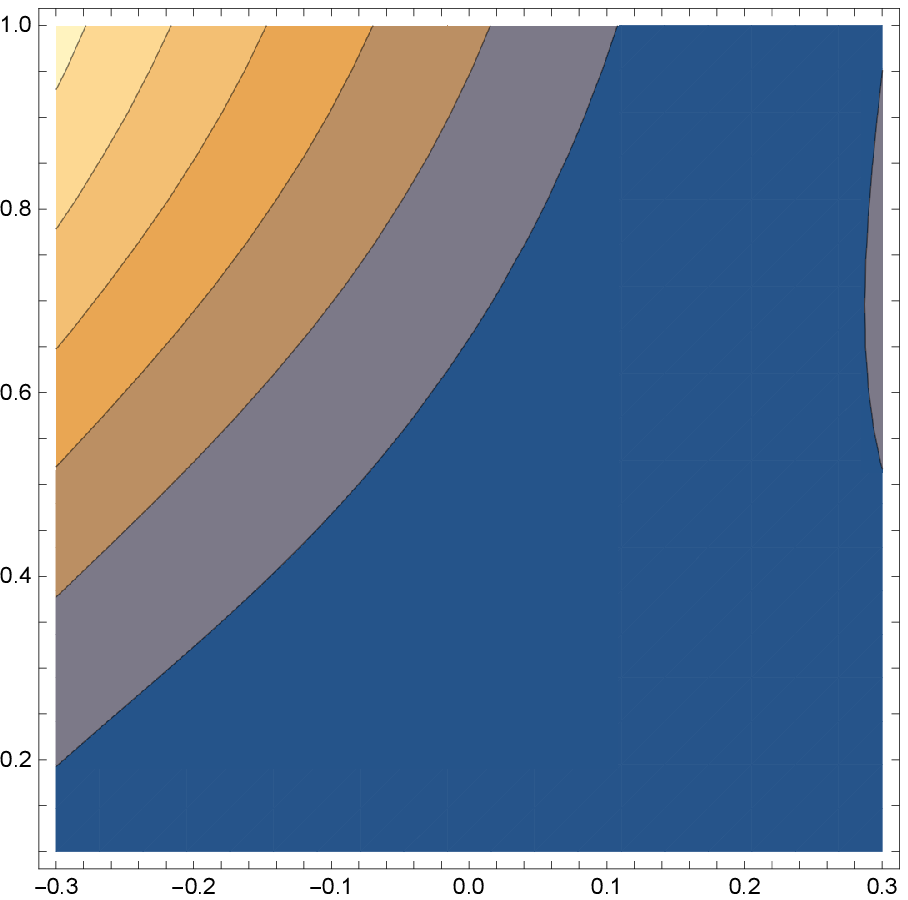}
\end{tabular}
\caption{
For the QOU model described in Section \ref{sec:examples}, we plot the absolute value of the relative error of our second order implied volatility approximation $|\sigb_2 - \sig|/\sig$ as a function of
$\log$-moneyness $k-x$ and caplet reset date $T$. The horizontal axis represents $\log$-moneyness $k -x$ and the vertical axis represents caplet reset date $T$. Ranging from
darkest to lightest, the regions above represent relative errors in increments of {$0.005$  from $< 0.005$ to $0.03$}. The initial date and settlement date of the caplet is fixed at $t = 0$ and $\Tb = 2$, respectively. The following parameters remained fixed: {$\kappa = 0.045$, {$\delta = \sqrt{0.035}$}, $y = \sqrt{0.08}$, $\theta = q = 0$}.}
\label{fig:qts-cir-err}
\end{figure}

\clearpage

% \blu{To Do}
% \\
% \red{$\Om$ was both sample space and argument of $F,G,H???$} \blu{That's ok}
% \\ 
% \red{$Y = (Y_t^{(1)}, Y_t^{(2)}, \ldots, Y_t^{(d)})^\top_{0\leq t \leq \Tb}$ should it be $t \geq 0$ or no subscript at all??}
% \\
% \red{Check all reference and citation}
% \\
% \red{Change a lot of $F,G,H$ eq. Check again.}
% \\
% \red{Affine short-rate vs. Affine term-structure}
% \\
% \red{$\rho = T-t$ ok?} \blu{NO! REMOVE ALL $\rho$}
% \\
% \red{Integral $\dd x$ all in back?} \blu{Yes do that}
% \\
% \red{Check that eq. trees are correct (used/unused)}
% \\
% \red{No equation gaps!}
% \\
% \red{New lines equation align +}
% \\
% \red{$G^2(t;T,\nu,\om)$ not $G(t;T,\nu,\om)^2$}
% \\
% \red{Short-rate , term-structure vs}
% \\
% \red{Vary sentence}
% \\
% \red{Don't start sentence with Eq}
% \\
% \red{Refer to term-structure by $(R,Y)$ or $R=Y$}
% \\
% \red{Make it different from previous paper!}
% \\
% \red{Make sure to add $1em$ to any new paragraph and no indentation I think??}

\end{document}